\let\LaTeXKernelLabel\label
\documentclass[aps,pra,reprint,superscriptaddress,nofootinbib]{revtex4-2}

\usepackage{amsmath,amssymb,amsfonts,amsthm,mathtools,bm}
\usepackage{microtype}
\usepackage{xcolor}
\makeatletter
\let\label\LaTeXKernelLabel
\let\ltx@label\label
\makeatother
\usepackage{hyperref}

\hypersetup{colorlinks=true,citecolor=blue!55!black,linkcolor=blue!55!black,urlcolor=blue!55!black,pdftitle={Analytic Maximal Violation of Extended MABK Inequalities for Generalized GHZ States}}

\newtheorem*{theorem*}{Theorem}
\newtheorem{theorem}{Theorem}

\newtheorem*{corollary*}{Corollary}
\newtheorem{lemma}{Lemma}
\newtheorem*{lemma*}{Lemma}

\newtheorem*{proposition*}{Proposition}

\newtheorem*{conjecture*}{Conjecture}
\theoremstyle{definition}
\newtheorem{definition}{Definition}
\newtheorem*{definition*}{Definition}
\theoremstyle{remark}
\newtheorem{remark}{Remark}
\newtheorem*{remark*}{Remark}

\newcommand{\Real}{\mathbb R}

\newcommand{\EquRef}[1]{Eq.~\eqref{#1}}

\newcommand{\ThmRef}[1]{\textbf{Theorem}~\ref{#1}}
\newcommand{\LemRef}[1]{\textbf{Lemma}~\ref{#1}}

\newcommand{\Tr}{\operatorname{Tr}}
\DeclareMathOperator{\diag}{diag}
\DeclareMathOperator{\rank}{rank}

\begin{document}

\title{Analytic Maximal Violation of Extended MABK Inequalities for Generalized GHZ States}

\author{Kun-Peng Wu}
\affiliation{School of Physics, Nankai University, Tianjin 300071, People's Republic of China}

\author{Hui-Xian Meng}
\affiliation{School of Mathematics and Physics, North China Electric Power University, Beijing 102206, China}

\author{Jie Zhou}
\affiliation{College of Physics and Materials Science, Tianjin Normal University, Tianjin
  300387, China}

\author{Jing-Ling Chen}
\email{chenjl@nankai.edu.cn}
\affiliation{Theoretical Physics Division, Chern Institute of Mathematics, Nankai University, Tianjin 300071, People's Republic of China}

\date{\today}

\begin{abstract}
  We analytically characterize the maximal quantum violation of the extended Mermin-Ardehali-Belinskii-Klyshko (EMABK) family of inequalities by $n$-qubit generalized Greenberger-Horne-Zeilinger (GHZ) states. We develop a correlation-tensor approach in which the Bell value is expressed as the  Frobenius inner product of the quantum correlation tensor and an effective coefficient tensor. A rank constraint on the latter, together with the von Neumann trace inequality, yields an upper bound governed by the two largest singular values of the reshaped correlation-tensor. For generalized GHZ states, we determine the complete singular spectrum and obtain a piecewise analytic upper bound. We then construct two complementary measurement strategies that saturate the bound for EMABK throughout the entire parameter range.
  The first is a purely MABK strategy in which all measurement directions lie in the equatorial plane of the Bloch sphere. The second is a hybrid strategy that recursively combines lower-order MABK anti-diagonal operators with the fully-$\sigma_z$ tensor product. The resulting piecewise analytic expression recovers the standard MABK Tsirelson bound $2^{(n-1)/2}$ in the maximally entangled limit and approaches the classical local-hidden-variable bound in the product-state limits. It further shows that every entangled generalized GHZ state exhibits a strict quantum-classical separation under the EMABK inequality, thereby eliminating the nonviolation region of the standard MABK inequality in the partially entangled regime.
\end{abstract}

\maketitle

\section{Introduction}

Bell nonlocality is a fundamental feature of quantum theory, revealing correlations that are incompatible with any local-realistic description~\cite{2014RMPBN, EPR1935}. Since Bell's theorem~\cite{Bell1964} and the subsequent Clauser-Horne-Shimony-Holt (CHSH) inequality~\cite{CHSH1969}, Bell inequalities have become indispensable tools for testing quantum nonlocality, with applications ranging from device-independent quantum key distribution~\cite{Ekert1991,Acin2006,Acin2007} to randomness certification~\cite{Pironio2010}. The degree of Bell violation provides a quantitative measure of nonlocal correlations, with the CHSH inequality attaining the Tsirelson bound $2\sqrt{2}$ for maximally entangled two-qubit states~\cite{Cirelson1980}. Experimental tests, from Aspect's pioneering experiments~\cite{Aspect1982} to recent loophole-free demonstrations~\cite{Hensen2015,Giustina2015,Shalm2015}, have firmly established the
violation of Bell inequalities.

In multipartite scenarios, the Mermin-Ardehali-Belinskii-Klyshko (MABK) family provides a natural generalization of the CHSH inequality, with quantum violations growing exponentially with the number of parties $n$~\cite{Mermin1990,GHZ1989,Ardehali1992,BelinskiiKlyshko1993}. The complete class of full-correlation Bell inequalities with two dichotomic observables per party was characterized by Werner and Wolf and independently by {\. Z}ukowski and Brukner~\cite{WernerWolf2001,ZukowskiBrukner2002}. MABK inequalities are powerful tools for detecting multipartite nonlocality and entanglement, and their maximal violations have also been exploited in the self-testing of Greenberger-Horne-Zeilinger (GHZ) states ~\cite{Kaniewski2016,Panwar2023}.


Previous theoretical studies have primarily focused on the maximal MABK violation $2^{(n-1)/2}$ attained by maximally entangled GHZ states at $\theta=\pi/4$. For partially entangled generalized GHZ states, however, the dependence of the maximal MABK value on the parameter $\theta$ has not been fully characterized analytically~\cite{ScaraniGisin2001}. Tailored Bell inequalities can nevertheless reveal the nonlocality of generalized three-qubit GHZ states throughout the entangled parameter range~\cite{ChenWuKwekOh2004}, and other two-setting multipartite families can detect states that are not detected by standard Bell inequalities~\cite{ChenAlbeverioFei2006}. More generally, every multipartite pure entangled state violates an appropriately constructed two-setting Bell inequality~\cite{YuEtAl2012}. The question addressed here is more specific: determining the exact maximal value within the recursively defined extended MABK (EMABK) family. This family has recently been constructed through recursive combinations of MABK operators~\cite{Fan2023} and remains violated throughout the partially entangled generalized-GHZ regime. Nevertheless, a systematic analytic expression for its maximal violation has remained unavailable.

In this paper, we provide a complete analytic characterization of the maximal quantum violation of the EMABK family for $n$-qubit generalized GHZ states.
We develop a unified correlation-tensor framework in which the quantum value of a Bell expression is written as the Frobenius inner product of the state-dependent correlation tensor $\mathcal{T}$ and the measurement-dependent effective coefficient tensor $\mathcal{C}$. By applying the von Neumann trace inequality together with rank constraints, we derive a tight upper bounds governed by the two largest singular values of the reshaped correlation tensor.
For generalized GHZ states, we determine the full singular spectrum and obtain a piecewise analytic bound. We then
construct two complementary measurement strategies that saturate this bound for EMABK throughout the entire parameter range: a purely equatorial MABK strategy and a hybrid strategy that combines transverse anti-diagonal correlations with the fully-$\sigma_z$ tensor product.
The result reaches $2^{(n-1)/2}$ at $\theta=\pi/4$, in agreement with the standard MABK result, and approaches the normalized classical bound of 1 as $\theta\to0$ or $\pi/2$.
Crucially, it eliminates the nonviolation region of the standard MABK inequality for partially entangled generalized GHZ states.


The rest of the paper is organized as follows. In Sec.~\ref{sec:tensor-framework}, we
introduce the correlation-tensor representation and derive the rank-constrained upper bound using the von Neumann trace inequality. In Sec.~\ref{sec:GHZ-MABKtype}, we determine the correlation tensor and its singular-value spectrum for generalized GHZ states and obtain the corresponding analytic bound. In Sec.~\ref{sec:EMABK}, we construct explicit measurement strategies that saturate the two branches of the bound and establish the exact maximal
EMABK violation. Section~\ref{sec:discussion} summarizes the main results and discusses their physical implications and possible extensions.

\section{Correlation-Tensor Framework and a Tight Upper Bound on Quantum Violation}\label{sec:tensor-framework}

Consider an $n$-qubit Bell scenario in which the $k$-th party has $m_k$ measurement settings, each yielding binary outcomes $\pm1$. An $(n,m,2)$-type full-correlation Bell inequality can be written as a linear combination of correlation functions:
\begin{equation}\label{eq:bell-inequality}
  \mathcal{B} = \sum_{x_1=0}^{m_1-1}\cdots\sum_{x_n=0}^{m_n-1} c_{x_1\cdots x_n}\, E(x_1,\dots,x_n) \le \mathcal{C}_{\text{LHV}},
\end{equation}
where $c_{x_1\cdots x_n}$ are real coefficients and $\mathcal{C}_{\text{LHV}}$ denotes the local-hidden-variable (LHV) bound. 
We assume throughout that the Bell expression has been normalized such that its local hidden-variable bound is unity, i.e., $\mathcal{C}_{\text{LHV}}=1$.

In the quantum setting, the parties share an $n$-qubit pure state $|\Psi\rangle$. For setting $x_k$, the dichotomic observable of the $k$-th party is specified by a unit Bloch vector $\vec{a}_{x_k}^{(k)}\in\mathbb{R}^3$ and the Pauli vector $\boldsymbol{\sigma}^{(k)}=(\sigma_x^{(k)},\sigma_y^{(k)},\sigma_z^{(k)})$:
\begin{equation}\label{eq:observable}
  A_{x_k}^{(k)} = \vec{a}_{x_k}^{(k)}\cdot\boldsymbol{\sigma}^{(k)}, \quad \bigl|\vec{a}_{x_k}^{(k)}\bigr|=1.
\end{equation}
The corresponding quantum correlation function is $E_Q(x_1,\dots,x_n)=\langle\Psi|A_{x_1}^{(1)}\otimes\cdots\otimes A_{x_n}^{(n)}|\Psi\rangle$.

To separate the dependence of the quantum value $\mathcal{B}_Q=\langle\Psi|\mathcal{B}|\Psi\rangle$ on the state from its dependence on the measurement settings, we employ the \emph{correlation-tensor} widely used in the analysis of full-correlation Bell inequalities~\cite{ZukowskiBrukner2002}.Define \emph{the $n$th-order correlation tensor} $\mathcal{T}$ with components
\begin{equation}\label{eq:T-def}
  T_{p_1\cdots p_n} = \bigl\langle\Psi\bigr| \sigma_{p_1}^{(1)}\otimes\cdots\otimes\sigma_{p_n}^{(n)} \bigl|\Psi\bigr\rangle, \quad p_k\in\{x,y,z\},
\end{equation}
where $\sigma_{p_k}^{(k)}$ denotes the Pauli operator acting on the $k$-th qubit. 
All tensor components are real because the Pauli tensor products are Hermitian. Importantly, $\mathcal{T}$ depends only on the quantum state and is independent of the Bell inequality and of the choice of local measurement directions.

In contrast, the Bell coefficients $c_{x_1\cdots x_n}$ and measurement directions $\vec{a}_{x_k}^{(k)}$ can be
combined into a measurement-dependent effective coefficient tensor $\mathbf{C}$,
\begin{equation}\label{eq:C-def}
  C_{p_1\cdots p_n} = \sum_{x_1,\dots,x_n} c_{x_1\cdots x_n}\, a_{x_1,p_1}^{(1)}\cdots a_{x_n,p_n}^{(n)}.
\end{equation}
Substituting \EquRef{eq:observable} into Eq.~\eqref{eq:bell-inequality}, and expanding in the Pauli basis, the quantum value of the Bell expression can be written as the Frobenius inner product of two real tensors of the same order:
\begin{equation}\label{eq:quantum-value}
  \mathcal{B}_Q = \sum_{p_1,\dots,p_n} C_{p_1\cdots p_n}\, T_{p_1\cdots p_n} \equiv  \langle\mathbf{C},\mathcal{T}\rangle_F.
\end{equation}
Flattening $\mathcal{T}$ and $\mathbf{C}$ into the $3^n$-dimensional real vectors $\vec{T}$ and $\vec{C}$, the Cauchy--Schwarz inequality immediately yields
\begin{equation}\label{eq:naive-bound}
  \mathcal{B}_Q = \vec{C}\cdot\vec{T} \le \|\mathcal{T}\|_F\,\|\mathbf{C}\|_F.
\end{equation}
This bound holds for arbitrary numbers of measurement settings but is generally not tight.

Our approach is conceptually analogous to the Horodecki CHSH criterion for two qubits~\cite{Horodecki1995} and the correlation-tensor framework developed by {\. Z}ukowski and Brukner~\cite{ZukowskiBrukner2002}. We tighten the bound by incorporating the rank constraints arising when at least one party has only two measurement settings.
Without loss of generality label this party as the first party, since tensor-index permutations preserve the Frobenius inner product. We thus take $m_1=2$, while other parties may adopt arbitrary setting numbers. Reshape the $n$th-order tensors $\mathcal T$ and $\mathcal C$ with respect to their first index to obtain real matrices $\widetilde{\mathcal T},\widetilde{\mathcal C}\in\mathbb R^{3\times 3^{n-1}}$ with entries
\begin{equation}\label{eq:reshape}
  \widetilde{\mathcal{T}}_{p_1,(p_2\cdots p_n)} = T_{p_1 p_2 \cdots p_n}, \quad
  \widetilde{\mathbf{C}}_{p_1,(p_2\cdots p_n)} = C_{p_1 p_2 \cdots p_n},
\end{equation}
where composite column indices $(p_2\cdots p_n)$ follow lexicographic ordering. The quantum expectation becomes a matrix Frobenius inner product
\begin{equation}
  \mathcal{B}_Q =\operatorname{Tr}\bigl(\widetilde{\mathcal{T}}^{\mathsf{T}}\widetilde{\mathbf{C}}\bigr) \equiv \langle\widetilde{\mathcal{T}},\widetilde{\mathbf{C}}\rangle_F .\label{eq:rewritten-quantum}
\end{equation}

\begin{theorem}\label{thm:tight-bound}
  Consider an $(n,m,2)$-type full-correlation Bell inequality, in which at least one party has only two measurement settings, and we label this party as party one. Then, for any $n$-qubit pure state $|\Psi\rangle$, the quantum value satisfies
  \begin{equation}\label{eq:the1}
    \mathcal{B}_Q \le \sqrt{s_1(\widetilde{\mathcal{T}})^2+s_2(\widetilde{\mathcal{T}})^2}\;\|\widetilde{\mathbf{C}}\|_F,
  \end{equation}
  where $s_1(\widetilde{\mathcal{T}})$ and $s_2(\widetilde{\mathcal{T}})$ denote the two largest singular values of the reshaped correlation matrix $\widetilde{\mathcal{T}}\in\mathbb{R}^{3\times3^{\,n-1}}$.
\end{theorem}

\begin{proof}
  From \EquRef{eq:rewritten-quantum}, the quantum value is $\mathcal{B}_Q=\operatorname{Tr}\bigl(\widetilde{\mathcal{T}}^{\mathsf{T}}\widetilde{\mathbf{C}}\bigr)$. Collect the two measurement directions of the party one as rows of a matrix $A^{(1)}\in\mathbb{R}^{2\times3}$, whose $i$-th row is $\vec{a}_i^{(1)\mathsf{T}}$. Using the definition of the effective coefficient tensor in \EquRef{eq:C-def}, the reshaped effective-coefficient matrix factorizes as
  \begin{equation}
    \widetilde{\mathbf{C}} = A^{(1)\mathsf{T}} M,
  \end{equation}
  where
  \begin{equation}
    M_{x_1,(p_2\cdots p_n)}=\sum_{x_2,\dots,x_n}c_{x_1x_2\cdots x_n}\,
    a_{x_2,p_2}^{(2)}\cdots a_{x_n,p_n}^{(n)}
  \end{equation}
  defines a $2\times3^{n-1}$ matrix depending only on the Bell coefficients and the settings of the remaining parties. Since $A^{(1)}$ possesses only two rows, $\rank(A^{(1)})\le2$, and hence
  \begin{equation}\label{eq:rank-constraint}
    \operatorname{rank}(\widetilde{\mathbf{C}})\le 2.
  \end{equation}
  This rank restriction holds irrespective of how many measurement settings are assigned to the remaining parties. Sort the singular values of $\widetilde{\mathcal T}$ and $\widetilde{\mathcal C}$ in non-increasing order. Because $\widetilde{\mathcal C}$ has rank at most two, only its first two singular values can be nonzero, and
  \begin{equation}\label{eq:C-singular-values}
    s_1(\widetilde{\mathbf{C}})^2 + s_2(\widetilde{\mathbf{C}})^2 = \|\widetilde{\mathbf{C}}\|_F^2.
  \end{equation}
  Applying the von Neumann trace inequality for real rectangular matrices (see Appendix~\ref{app:von Neumann trace inequality})
  \begin{equation}\label{eq:vn-trace}
    \operatorname{Tr}\bigl(\widetilde{\mathcal{T}}^{\mathsf{T}}\widetilde{\mathbf{C}}\bigr)
    \le s_1(\widetilde{\mathcal{T}})\,s_1(\widetilde{\mathbf{C}}) + s_2(\widetilde{\mathcal{T}})\,s_2(\widetilde{\mathbf{C}}).
  \end{equation}
  And applying the Cauchy--Schwarz inequality to these two terms on the right-hand side gives
  \begin{multline}
    s_1(\widetilde{\mathcal{T}})s_1(\widetilde{\mathbf{C}})
    +s_2(\widetilde{\mathcal{T}})s_2(\widetilde{\mathbf{C}})\\
    \le
    \sqrt{s_1(\widetilde{\mathcal{T}})^2+s_2(\widetilde{\mathcal{T}})^2}\,
    \sqrt{s_1(\widetilde{\mathbf{C}})^2+s_2(\widetilde{\mathbf{C}})^2}.
  \end{multline}
  Combining the above inequalities recovers Eq.~\EquRef{eq:the1}.
  Thus, the state-dependent factor in the upper bound is fully determined by the two dominant singular values of the reshaped correlation tensor, while the rank constraint restricts the effective coefficient matrix to at most two nonzero singular values.
\end{proof}

\section{Generalized GHZ States and MABK-Type Inequalities}
\label{sec:GHZ-MABKtype}

Consider the generalized $n$-qubit GHZ state
\begin{equation}
  \label{eq:ghz-general}
  |\Psi_\theta^{(n)}\rangle = \cos\theta\,|0\rangle^{\otimes n} + \sin\theta\,|1\rangle^{\otimes n},
  \quad \theta\in[0,\tfrac{\pi}{2}].
\end{equation}

\begin{theorem}[Correlation-tensor structure of generalized GHZ states]\label{thm:ghz-tensor}
  For the $n$-qubit generalized GHZ state \EquRef{eq:ghz-general}, the only nonzero components of its $n$th-order correlation tensor $\mathcal{T}$ belong to the following two families:
  \begin{enumerate}
    \item[(i)] the fully-$z$ component,
      $T_{zz\dots z}=\cos^2\theta+(-1)^n\sin^2\theta$;
    \item[(ii)] components without any $z$ indices and containing an even number $n_y$ of $y$ indices,
      $T_{p_1\dots p_n}=\sin(2\theta)(-1)^{n_y/2}$, where $p_k\in\{x,y\}$.
  \end{enumerate}
  After reshaping $\mathcal{T}$ along its first index into $\widetilde{\mathcal{T}}\in \mathbb R^{3\times 3^{n-1}}$, the singular-value multiset reads
  \begin{equation}\label{eq:singular-values}
    \begin{aligned}
      \lambda_\perp & =2^{\frac{n-2}{2}}\sin 2\theta
                    &                                              & \text{(multiplicity two)}, \\
      \lambda_z     & =\bigl|\cos^2\theta+(-1)^n\sin^2\theta\bigr|
                    &                                              & \text{(multiplicity one)}.
    \end{aligned}
  \end{equation}
\end{theorem}

\begin{proof}
  Substituting \EquRef{eq:ghz-general} into the expectation value of a Pauli tensor product and expanding in the computational basis. Nonzero contributions to $T_{p_1\dots p_n}$ originate only from the two index classes listed above.

  For the all-$z$ component, $\langle0|\sigma_z|0\rangle=1$ and $\langle1|\sigma_z|1\rangle=-1$, we obtain
  \begin{equation}
    \label{eq:T-n-zzz}
    T_{zz\dots z} = \cos^2\theta + (-1)^n\sin^2\theta,
  \end{equation}
  which equals $1$ for even $n$ and $\cos2\theta$ for odd $n$. The second class consists of components with no $z$ indices and an even number of $y$ indices, the real part of the off-diagonal contribution is nonzero and gives
  \begin{equation}
    \label{eq:T-n-xy}
    T_{p_1\dots p_n} = \sin(2\theta)\,(-1)^{n_y/2},
    \quad p_k\in\{x,y\}.
  \end{equation}
  We now reshape $T_{p_1\dots p_n}$ along its first index into the matrix $\widetilde{\mathcal{T}}\in\mathbb{R}^{3\times3^{n-1}}$. The $x$ and $y$ rows contain only components from the $xy$ sector, whereas the $z$ row contains only the all-$z$ component. Moreover, the supports of the $x$ and $y$ rows are disjoint because the remaining indices contain, respectively, an even and an odd number of $y$ indices. The three rows are therefore mutually orthogonal, so the Gram matrix $\widetilde{\mathcal{T}}\widetilde{\mathcal{T}}^{\mathsf T}$ is diagonal.

  Each of the $x$ and $y$ rows has $2^{n-2}$ nonzero entries of magnitude $\sin2\theta$, while the $z$ row has the single entry given in \EquRef{eq:T-n-zzz}. Hence the row norms, and thus the singular values, are
  \begin{equation}
    \label{eq:T-singular}
    \lambda_\perp=2^{\frac{n-2}{2}}\sin2\theta,\qquad
    \lambda_z=\bigl|\cos^2\theta+(-1)^n\sin^2\theta\bigr|.
  \end{equation}
  Here $\lambda_\perp$ has multiplicity two. For even $n$, $\lambda_z=1$; for odd $n$, $\lambda_z=|\cos2\theta|$.
\end{proof}

The MABK family is defined recursively. For bipartite ($n=2$), the normalized MABK Bell operator coincides with the rescaled CHSH operator
\begin{equation}
  \label{eq:MABK-rec-2}
  \mathcal{B}_2 = \frac{1}{2}(M_1^{(1)}M_1^{(2)} + M_1^{(1)}M_2^{(2)} + M_2^{(1)}M_1^{(2)} - M_2^{(1)}M_2^{(2)}),
\end{equation}
where $M_1^{(k)}$, $M_2^{(k)}$ denote the two dichotomic observables of the $k$-th party, $k=1, 2$. For $n\ge3$,
\begin{equation}
  \label{eq:MABK-rec-n}
  \mathcal{B}_n = \frac{1}{2}\bigl[\mathcal{B}_{n-1}^{(++)}(M_1^{(n)}+M_2^{(n)})
  + \mathcal{B}_{n-1}^{(+-)}(M_1^{(n)}-M_2^{(n)})\bigr],
\end{equation}
where $\mathcal{B}_{n-1}^{(++)}$ is the $(n-1)$-party MABK polynomial and $\mathcal{B}_{n-1}^{(+-)}$ is obtained from it by interchanging the two measurement settings, $1\leftrightarrow2$ for every site.

The EMABK family is constructed similarly but couples two independently parametrized lower-order MABK operators. The $n$-party EMABK operator reads
\begin{equation}
  \label{eq:EMABK-rec}
  \mathcal{B}_{n}^{\text{EMABK}} = \frac{1}{2}\bigl[
  \mathcal{B}_{n-1}^{(1)}(M_1^{(n)}+M_2^{(n)})
  + \mathcal{B}_{n-1}^{(2)}(M_1^{(n)}-M_2^{(n)})\bigr],
\end{equation}
where $\mathcal{B}_{n-1}^{(1)}$ and $\mathcal{B}_{n-1}^{(2)}$ are $(n-1)$-party MABK polynomials with independent measurement directions. We take $\mathcal{B}_{n-1}^{(1)}$ to involve settings 1 and 2 at each party, while $\mathcal{B}_{n-1}^{(2)}$ is obtained by replacing these settings with settings 3 and 4, respectively, at every participating party.
Despite their different coupling structures, the MABK and EMABK Bell polynomials share the same fundamental algebraic properties that make them amenable to the tensor framework of Section~\ref{sec:tensor-framework}.

\begin{lemma}\label{lem:mabk-emabk-norm}
  The MABK Bell polynomial $\mathcal{B}_n$ defined by Eqs.~\eqref{eq:MABK-rec-2} and \eqref{eq:MABK-rec-n}, and the EMABK Bell polynomial $\mathcal{B}_{n}^{\text{EMABK}}$ defined by \EquRef{eq:EMABK-rec}, satisfy
  \begin{equation}
    \|\mathbf{C}\|_F^2 = 1.
  \end{equation}
\end{lemma}

Proofs are given in Appendix~\ref{app:mabk-norm}. Their reshaped effective-coefficient tensors also satisfy the low-rank condition $\rank(\widetilde{\mathbf C})\le2$ (Appendix C). Therefore our general upper bound Theorem 1 \EquRef{eq:the1} applies to both families. Inserting the singular values from Theorem 2 \EquRef{eq:singular-values} yields the piecewise upper bound below.

\begin{theorem}\label{thm:unified-bound}
  For the generalized GHZ state $|\Psi_\theta^{(n)}\rangle$, the quantum values of both the MABK and EMABK Bell polynomials satisfy
  \begin{equation}
    \label{eq:MABK-theoretical-limit}
    \mathcal{B}_Q \le
    \begin{cases}
      2^{\frac{n-1}{2}}\sin2\theta,
       & \lambda_\perp\ge\lambda_z, \\[8pt]
      \sqrt{\lambda_z^2+2^{n-2}\sin^2 2\theta},
       & \lambda_\perp\le\lambda_z,
    \end{cases}
  \end{equation}
  where
  \begin{equation}
    \lambda_\perp=2^{\frac{n-2}{2}}\sin2\theta,\qquad
    \lambda_z=\bigl|\cos^2\theta+(-1)^n\sin^2\theta\bigr|.\nonumber
  \end{equation}
\end{theorem}

The first branch $2^{(n-1)/2}\sin2\theta$ governed by the two degenerate transverse singular values $\lambda_\perp$ and corresponds to measurement directions confined to the $xy$ equatorial plane. At $\theta=\pi/4$, this branch yields the well-known maximal MABK quantum value $2^{(n-1)/2}$; it vanishes approaching product-state limits $\theta\to0,\pi/2$. The second branch combines one transverse singular value with $\lambda_z$ and approaches the normalized classical bound of $1$ in the product-state limits. It therefore exploits the fully-$z$ component $T_{zz\dots z}$ to maintain a nonzero Bell value.

Although the same mathematical upper bound \EquRef{eq:MABK-theoretical-limit} holds for both families, they differ fundamentally in its attainability. The standard MABK structure imposes additional constraints on the effective coefficients: its effective coefficient matrix $\widetilde{\mathbf C}$ cannot be continuously adjusted with $\theta$ to simultaneously align with the transverse and $z$-directional components of $\widetilde{\mathcal T}$. Its optimal violation is restricted to the first branch $2^{(n-1)/2}\sin2\theta$. In the weakly-entangled regime, where $\theta$ approaches $0$ or $\pi/2$, this value may drop below the classical LHV bound, producing a non-violation region. For the standard MABK family with $n\ge3$, the second branch of \EquRef{eq:MABK-theoretical-limit} is unattainable away from the maximally-entangled point.
By contrast, the two independent sets of measurement directions in the EMABK construction allow it to exploit the full singular-value structure of the correlation tensor and saturate both branches of \EquRef{eq:MABK-theoretical-limit}. The EMABK family can thus attain the second-branch violation that is inaccessible to the standard MABK family, eliminating the nonviolation region for nonmaximally entangled generalized GHZ states. Figure~\ref{fig:emabk-vs-mabk-n3} illustrates this distinction for $n=3$.

\begin{figure}[t]
  \centering
  \includegraphics[width = 1\columnwidth]{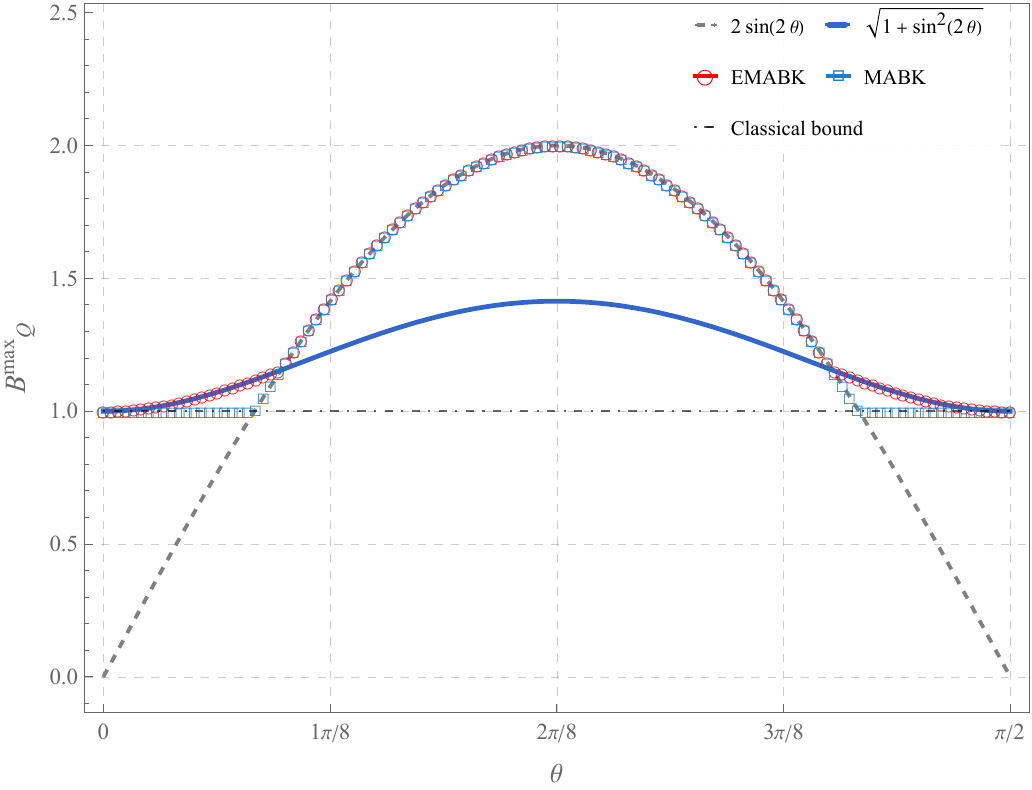}
  \caption{Maximal quantum values of the EMABK and standard MABK expressions for the three-qubit state $\cos\theta|000\rangle+\sin\theta|111\rangle$, with $\theta\in[0,\pi/2]$. Red open circles and blue open squares show numerical optimization results for EMABK and MABK, respectively. The gray dashed and blue solid curves show the corresponding analytic results, and the black dash-dotted line marks the classical bound. The standard MABK expression has a nonviolation region in the weakly entangled regime ($\theta\lesssim0.26$), whereas the EMABK expression violates the classical bound throughout the entangled parameter range by means of a hybrid measurement strategy.}
  \label{fig:emabk-vs-mabk-n3}
\end{figure}

\section{Measurement Strategies Saturating the EMABK Bound}
\label{sec:EMABK}

Both branches of the piecewise upper bound in \EquRef{eq:MABK-theoretical-limit} can be saturated by the EMABK Bell polynomial through suitable choices of measurement directions. In this section, we construct explicit measurement settings for the two optimal strategies, thereby proving that the bound is tight for EMABK over the entire parameter range.

It is worth emphasizing that the EMABK construction in \EquRef{eq:EMABK-rec} inherently expands the optimization space by allowing two independent $(n-1)$-party MABK blocks, $\mathcal{B}_{n-1}^{(1)}$ and $\mathcal{B}_{n-1}^{(2)}$, to employ distinct measurement settings. Consequently, the global maximization over the state parameter $\theta$ may naturally lead to degenerate cases where some settings of the first $n-1$ parties coincide. These degeneracies do not alter the Bell coefficients or the classical bounds but restrict the effective-coefficient tensor to a lower-dimensional subspace.

We parametrize the unit Bloch vector $\hat{m}_i^{(k)}$ associated with the observable $M_i^{(k)}=\hat{m}_i^{(k)}\!\cdot\!\bm{\sigma}^{(k)}$ by the spherical angles $(\vartheta_i^{(k)},\varphi_i^{(k)})$:
\begin{equation}\label{eq:spherical}
  \hat{m}_i^{(k)} = \bigl(\sin\vartheta_i^{(k)}\cos\varphi_i^{(k)},\;
  \sin\vartheta_i^{(k)}\sin\varphi_i^{(k)},\;
  \cos\vartheta_i^{(k)}\bigr),
\end{equation}
where $\vartheta_i^{(k)}\in[0,\pi]$ is the polar angle and $\varphi_i^{(k)}\in[0,2\pi)$ is the azimuthal angle. 

\textbf{Strategy I.}
The first branch saturation value $2^{(n-1)/2}\sin2\theta$ is attained by the standard $n$-party MABK operator. All measurement polar angles are set to $\pi/2$ and azimuthal angles are chosen to cancel phases of antidiagonal matrix elements.
For parties $k=1,\dots,n-1$, which possess four observables $\{M_1^{(k)},M_2^{(k)},M_3^{(k)},M_4^{(k)}\}$ within EMABK construction, we set
\begin{equation}
  \begin{aligned}
    M_1^{(k)} & :\; \vartheta_{1}^{(k)}=\frac{\pi}{2},\; \varphi_{1}^{(k)} = -\frac{(n-1)\pi}{4n},               \\
    M_2^{(k)} & :\; \vartheta_{2}^{(k)}=\frac{\pi}{2},\; \varphi_{2}^{(k)} = -\frac{(n-1)\pi}{4n}+\frac{\pi}{2}, \\
    M_3^{(k)} & := M_2^{(k)},\quad M_4^{(k)}:= M_1^{(k)}.
  \end{aligned}
\end{equation}
Here $M_3^{(k)}$ and $M_4^{(k)}$ are identified with $M_2^{(k)}$ and $M_1^{(k)}$, respectively, so that $\mathcal{B}_{n-1}^{(2)}$ is exactly the $1\leftrightarrow 2$ interchange of $\mathcal{B}_{n-1}^{(1)}$ (the standard MABK partner operator). Thus the EMABK polynomial in \EquRef{eq:EMABK-rec} reduces to the standard MABK polynomial in \EquRef{eq:MABK-rec-n}.

For the $n$-th party with only two observables $\{M_1^{(n)},M_2^{(n)}\}$, we choose
\begin{align}
  M_1^{(n)}:\; \vartheta_{1}^{(n)}=\frac{\pi}{2},\; \varphi_{1}^{(n)} = -\frac{(n-1)\pi}{4n}, \\
  M_2^{(n)}:\; \vartheta_{2}^{(n)}=\frac{\pi}{2},\; \varphi_{2}^{(n)} = -\frac{(n-1)\pi}{4n}+\frac{\pi}{2}.
\end{align}
With these settings, the EMABK polynomial in \EquRef{eq:EMABK-rec} reduces to the standard MABK polynomial in \EquRef{eq:MABK-rec-n}, and its expectation value on $|\Psi_\theta^{(n)}\rangle$ is
\begin{equation}\label{eq:strat1-val}
  \langle\mathcal{B}_n^{\mathrm{EMABK}}\rangle = 2^{\frac{n-1}{2}}\sin 2\theta,
\end{equation}
which coincides with the first branch of \EquRef{eq:MABK-theoretical-limit}. This strategy is optimal when $\lambda_\perp=2^{(n-2)/2}\sin2\theta\ge\lambda_z$ and it recovers the maximal MABK value $2^{(n-1)/2}$ at $\theta=\pi/4$.

\textbf{Strategy II.}
The second branch saturation value $\sqrt{\lambda_z^2+2^{n-2}\sin^2 2\theta}$ is attained by a hybrid construction that exploits the two independently parametrized lower-order operators in \EquRef{eq:EMABK-rec}. The first $n-1$ parties supply both $xy$-plane anti-diagonal structure through $\mathcal{B}_{n-1}^{(1)}$ and fully-$z$ diagonal structure through $\mathcal{B}_{n-1}^{(2)}$, while the observables of the $n$-th party are optimized for given $\theta$.

For $k=1,\dots,n-1$:
\begin{equation}
  \begin{aligned}
    M_1^{(k)} & :\; \vartheta_{1}^{(k)}=\frac{\pi}{2},\; \varphi_{1}^{(k)} = -\frac{(n-2)\pi}{4(n-1)},               \\
    M_2^{(k)} & :\; \vartheta_{2}^{(k)}=\frac{\pi}{2},\; \varphi_{2}^{(k)} = -\frac{(n-2)\pi}{4(n-1)}+\frac{\pi}{2}, \\
    M_3^{(k)} & = M_4^{(k)} \;:\; \vartheta_{3,4}^{(k)}=0,\; \varphi_{3,4}^{(k)}=0.
  \end{aligned}
\end{equation}
Due to $M_3^{(k)}=M_4^{(k)}=\sigma_z$, the second lower-order operator reduces to $\mathcal{B}_{n-1}^{(2)}=\sigma_z^{\otimes(n-1)}$. Consequently, the EMABK operator \EquRef{eq:EMABK-rec} simplifies to
\begin{align}\label{eq:EMABK-degenerate-II}
  \mathcal{B}_{n}^{\mathrm{EMABK}}= & \frac{1}{2}\Big[\mathcal{B}_{n-1}^{(1)}\otimes\big(M_1^{(n)}+M_2^{(n)}\big)\nonumber \\
                                    & +{M_3^{(n)}}^{\otimes(n-1)}\otimes\big(M_1^{(n)}-M_2^{(n)}\big)\Big],
\end{align}
where $M_3^{(k)}=\sigma_z, k=1,2,\dots,n$ and $\mathcal{B}_{n-1}^{(1)}$ is the standard $(n-1)$-party MABK operator with purely anti-diagonal structure. It is noteworthy that, although the EMABK framework assigns four settings to each of the first $n-1$ parties, the optimal strategy in this branch effectively uses only three independent physical directions per party: two orthogonal directions within the $xy$ plane and one $z$ direction.

For a given number of parties $n$ and state parameter $\theta$, define
\begin{align}\label{eq:ABR}
  A_n(\theta) & =\cos^2\theta+(-1)^n\sin^2\theta,\quad\nonumber  \\
  B_n(\theta) & =2^{\frac{n-2}{2}}\sin\theta\cos\theta,\nonumber \\
  R_n(\theta) & =\sqrt{\frac{A_n(\theta)^2}{4}+B_n(\theta)^2}.
\end{align}
For the $n$-th party, the two optimal measurement directions lie in the $xz$ plane ($\varphi=0$), sharing equal $x$ components and opposite $z$ components:
\begin{equation}
  \begin{aligned}
    M_1^{(n)} & :\; \vartheta_{1}^{(n)} = \arccos\!\Bigl(\frac{A_n(\theta)}{2R_n(\theta)}\Bigr),\quad \varphi_{1}^{(n)}=0, \\[4pt]
    M_2^{(n)} & :\; \vartheta_{2}^{(n)} = \pi - \vartheta_{1}^{(n)},\quad \varphi_{2}^{(n)}=0.
  \end{aligned}
\end{equation}
With these settings, the expectation value of the EMABK polynomial in \EquRef{eq:EMABK-rec} is
\begin{equation}\label{eq:strat2-val}
  \langle\mathcal{B}_n^{\mathrm{EMABK}}\rangle
  = \sqrt{\bigl[\cos^{2}\theta+(-1)^{n}\sin^{2}\theta\bigr]^{2}+2^{n-2}\sin^{2}2\theta},
\end{equation}
which is exactly the second branch of \EquRef{eq:MABK-theoretical-limit}. For every $\theta\in(0,\pi/2)$, this value is strictly exceeds the normalized classical bound $1$, including the weakly-entangled regimes where standard MABK expression is not violated (see Fig.~\ref{fig:emabk-vs-mabk-n3}). Thus, the additional degrees of freedom provided by the two independent lower-order branches enable the EMABK inequality to eliminate the nonviolation region of the standard MABK inequality.

\begin{figure}[t]
  \centering
  \includegraphics[width=\columnwidth]{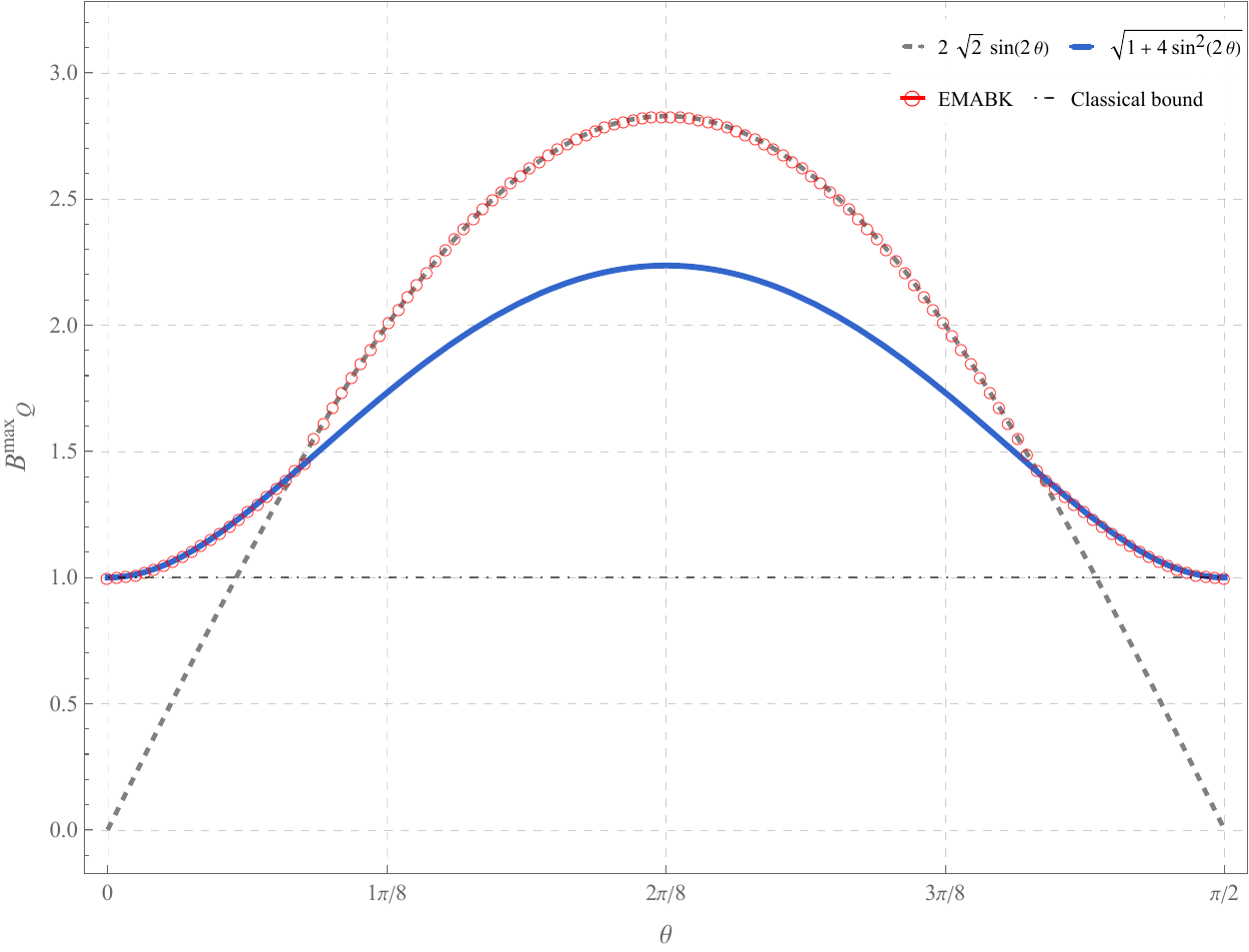}
  \caption{Maximal quantum values of the EMABK inequality for the four-qubit generalized GHZ state $\cos\theta|0000\rangle+\sin\theta|1111\rangle$ with $\theta\in[0,\pi/2]$. The gray dashed curve shows the first branch $2\sqrt{2}\sin2\theta$ attained by Strategy~I (purely equatorial MABK measurements), and the blue solid curve shows the second branch $\sqrt{1+4\sin^2 2\theta}$ attained by Strategy~II (hybrid measurements combining $xy$-plane operators with $\sigma_z$). Red open circles show numerical optimization results, and the black dash-dotted line marks the classical bound. The two branches cross at $\theta=\pi/12$.}
  \label{fig:emabk-n4-violation}
\end{figure}

\begin{figure}[t]
  \centering
  \includegraphics[width=\columnwidth]{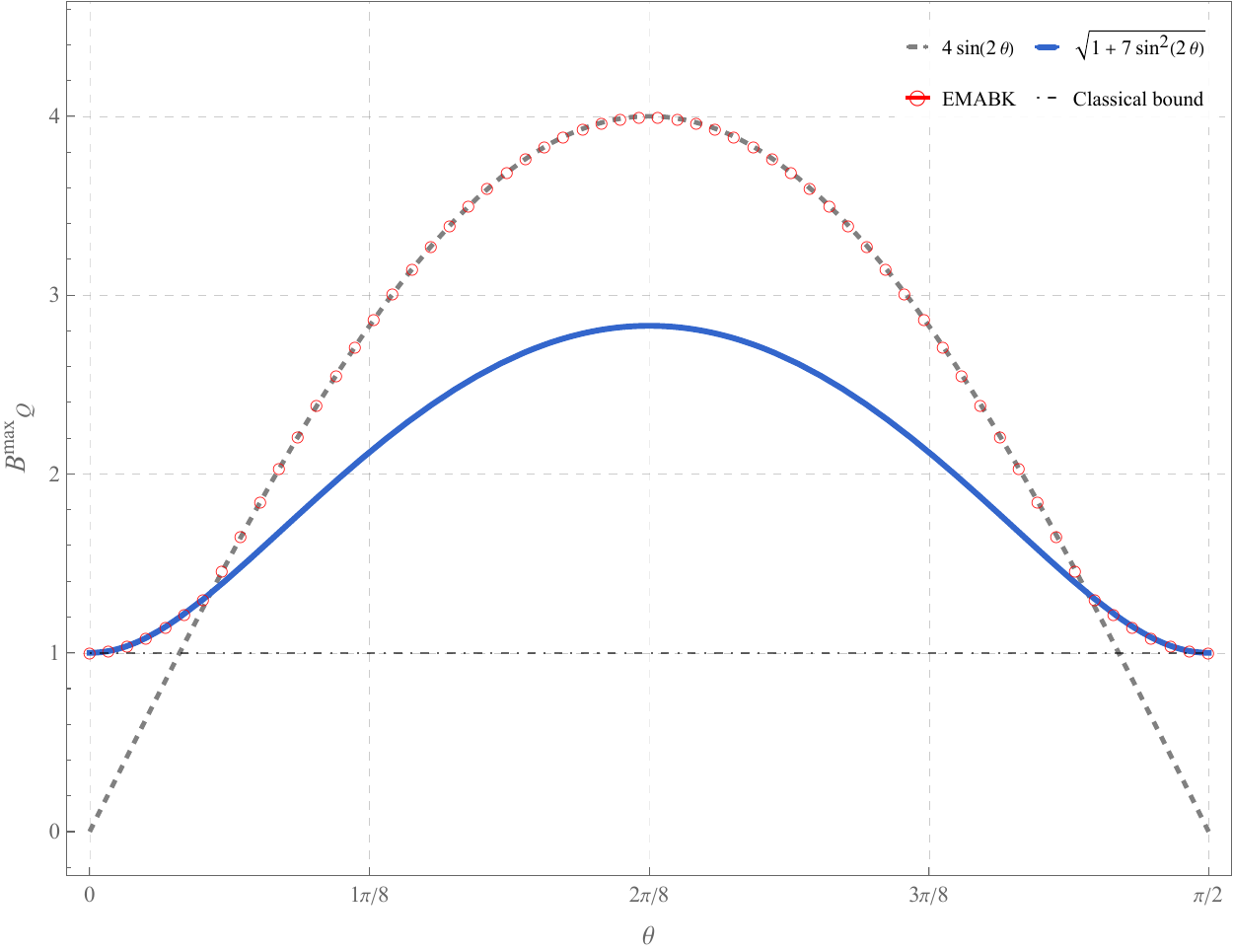}
  \caption{Maximal quantum values of the EMABK inequality for the five-qubit generalized GHZ state $\cos\theta|00000\rangle+\sin\theta|11111\rangle$ with $\theta\in[0,\pi/2]$. The gray dashed curve shows the first branch $4\sin2\theta$ attained by Strategy~I, and the blue solid curve shows the second branch $\sqrt{1+7\sin^2 2\theta}$ attained by Strategy~II. Red open circles show numerical optimization results, and the black dash-dotted line marks the classical bound. The two branches meet at $\sin2\theta=1/3$ ($\theta\approx0.17$).}
  \label{fig:emabk-n5-violation}
\end{figure}

\emph{Example ($n=4$).} In Strategy~I all four parties measure in the $xy$ plane with azimuthal angles $\varphi_1^{(k)}=-3\pi/16$ and $\varphi_2^{(k)}=-3\pi/16+\pi/2$ ($k=1,\dots,4$), while for the first three parties the auxiliary settings are identified as $M_3^{(k)}=M_2^{(k)}$ and $M_4^{(k)}=M_1^{(k)}$. The EMABK operator reduces to the standard MABK polynomial, yielding the first-branch value $2\sqrt{2}\sin 2\theta$. In Strategy~II the first three parties take $\varphi_1^{(k)}=-\pi/6$ and $\varphi_2^{(k)}=-\pi/6+\pi/2$ with $M_3^{(k)}=M_4^{(k)}=\sigma_z$. Since $n$ is even, $A_4(\theta)=1$ and the optimal polar angle for the fourth party satisfies
\begin{equation}
  \cos\vartheta_1^{(4)}=\frac{1}{\sqrt{1+4\sin^2 2\theta}},
\end{equation}
giving the second-branch value $\sqrt{1+4\sin^2 2\theta}$. The two branches meet at $\theta=\pi/12$, where $2\sin 2\theta=1=\lambda_z$.

\emph{Example ($n=5$).} Strategy~I sets $\varphi_1^{(k)}=-\pi/5$ for all five parties, producing the first-branch value $4\sin 2\theta$. For Strategy~II the first four parties choose $\varphi_1^{(k)}=-3\pi/16$ and $M_3^{(k)}=M_4^{(k)}=\sigma_z$, while the odd-$n$ diagonal factor $A_5(\theta)=\cos 2\theta$ gives
\begin{equation}
  \cos\vartheta_1^{(5)}=\frac{\cos 2\theta}{\sqrt{1+7\sin^2 2\theta}}
\end{equation}
for the fifth party. The resulting second-branch value $\sqrt{1+7\sin^2 2\theta}$. The crossing point occurs at $2^{3/2}\sin 2\theta=|\cos 2\theta|$, i.e.\ at $\sin 2\theta=1/3$ ($\theta\approx0.17$).

Figures~\ref{fig:emabk-n4-violation} and \ref{fig:emabk-n5-violation} provide further numerical verification for $n=4$ and $5$. Numerical-optimization data (red open circles) match our piecewise analytical bound across the full $\theta$ interval.
Taken together, the two measurement strategies prove that the piecewise upper bound \EquRef{eq:MABK-theoretical-limit} is attained by the EMABK Bell polynomial for every $\theta\in[0,\pi/2]$. The transition point is given by the condition $\lambda_\perp=\lambda_z$. Switching between the two sets of measurement directions at this point yields a continuous maximal-value curve that equals the classical bound only at the product-state limits and exceeds it everywhere else.

We remark that neither optimal strategy requires four mutually non-commuting observables for first $n-1$ parties: Strategy~I effectively employs only two independent directions (the MABK-equivalent degeneracy). Strategy II uses three directions per site: two orthogonal directions in the $xy$ plane plus $\sigma_z$.
This confirms that the advantage of EMABK is not from adding  genuinely independent noncommuting observables per se, but from the enlarged linear-combination space of two independent lower-order MABK structures, which permits adaptive switching between complementary correlation sectors.

\section{Discussion and Conclusions}\label{sec:discussion}

In this work, we have developed a correlation-tensor framework for analyzing multipartite full-correlation Bell inequalities with a two-setting party. The rank constraint of the corresponding effective coefficient matrix reduces the optimization to the two largest singular values of the reshaped correlation tensor. For generalized GHZ states, this structure leads to two degenerate transverse singular values and one longitudinal singular value, providing a simple geometric picture of the Bell optimization. Applying this framework to the MABK and EMABK families, we obtain a common singular-value upper bound. We further construct explicit measurement strategies that saturate this bound for EMABK over the entire parameter range. The conventional MABK strategy accesses the transverse correlation sector, whereas the additional freedom of EMABK allows one transverse mode
to be combined with the longitudinal correlation. This difference explains why EMABK can exhibit Bell violation in the parameter region where the standard MABK inequality does not.

Our results show that the enhancement of EMABK does not simply arise from using more measurement settings. Instead, the independent lower-order MABK blocks enlarge the accessible coefficient-tensor space and allow different
correlation sectors of the quantum state to be combined. This provides a geometric interpretation of the enhanced nonlocality and establishes the EMABK bound as an exact quantum maximum for the generalized GHZ states considered here. The present analysis focuses on pure generalized GHZ states and projective qubit measurements. Extending the singular-value approach to noisy states, other multipartite entangled states, and experimentally simpler measurement configurations would be natural directions for future work.
More broadly, the correlation-tensor perspective may provide a useful route for identifying optimal measurement structures and designing multipartite Bell inequalities adapted to specific quantum correlations.

\begin{acknowledgments}
  This work is supported by the Quantum Science and Technology-National Science and Technology Major Project (Grant No. 2024ZD0301000), and the National Natural Science Foundation of China (Grant No. 12275136). H. X. Meng was supported by the Beijing Natural Science Foundation (Grant No. 1262039), and the National Natural Science Foundations of China(Grant No. 12461087). J. Z. was supported by Science \& Technology Development Fund of Tianjin Education Commission for Higher Education(No. 2024KJ059).

\end{acknowledgments}

\bibliography{references}

@article{2014RMPBN,
  title   = {Bell Nonlocality},
  author  = {Brunner, Nicolas and Cavalcanti, Daniel and Pironio, Stefano and Scarani, Valerio and Wehner, Stephanie},
  year    = {2014},
  month   = apr,
  journal = {Rev. Mod. Phys.},
  volume  = {86},
  number  = {2},
  pages   = {419--478},
  issn    = {0034-6861, 1539-0756},
  doi     = {10.1103/RevModPhys.86.419},
  url     = {https://link.aps.org/doi/10.1103/RevModPhys.86.419},
  urldate = {2022-05-29},
  langid  = {english}
}

@article{EPR1935,
  title     = {Can Quantum-Mechanical Description of Physical Reality Be Considered Complete?},
  author    = {Einstein, A. and Podolsky, B. and Rosen, N.},
  journal   = {Phys. Rev.},
  volume    = {47},
  issue     = {10},
  pages     = {777--780},
  numpages  = {0},
  year      = {1935},
  month     = {May},
  publisher = {American Physical Society},
  doi       = {10.1103/PhysRev.47.777},
  url       = {https://link.aps.org/doi/10.1103/PhysRev.47.777}
}

@article{Bell1964,
  title     = {On the {Einstein} {Podolsky} {Rosen} paradox},
  author    = {Bell, J. S.},
  journal   = {Physics Physique Fizika},
  volume    = {1},
  issue     = {3},
  pages     = {195--200},
  numpages  = {6},
  year      = {1964},
  month     = {Nov},
  publisher = {American Physical Society},
  doi       = {10.1103/PhysicsPhysiqueFizika.1.195},
  url       = {https://link.aps.org/doi/10.1103/PhysicsPhysiqueFizika.1.195}
}

@article{CHSH1969,
  title     = {Proposed Experiment to Test {Local} {Hidden}-{Variable} Theories},
  author    = {Clauser, John F. and Horne, Michael A. and Shimony, Abner and Holt, Richard A.},
  journal   = {Phys. Rev. Lett.},
  volume    = {23},
  issue     = {15},
  pages     = {880--884},
  numpages  = {0},
  year      = {1969},
  month     = {Oct},
  publisher = {American Physical Society},
  doi       = {10.1103/PhysRevLett.23.880},
  url       = {https://link.aps.org/doi/10.1103/PhysRevLett.23.880}
}

@article{Ekert1991,
  title     = {Quantum cryptography based on {Bell}'s theorem},
  author    = {Ekert, Artur K.},
  journal   = {Phys. Rev. Lett.},
  volume    = {67},
  issue     = {6},
  pages     = {661--663},
  numpages  = {0},
  year      = {1991},
  month     = {Aug},
  publisher = {American Physical Society},
  doi       = {10.1103/PhysRevLett.67.661},
  url       = {https://link.aps.org/doi/10.1103/PhysRevLett.67.661}
}

@article{Acin2006,
  title     = {From {Bell}'s Theorem to Secure Quantum Key Distribution},
  author    = {Ac\'{\i}n, Antonio and Gisin, Nicolas and Masanes, Lluis},
  journal   = {Phys. Rev. Lett.},
  volume    = {97},
  issue     = {12},
  pages     = {120405},
  numpages  = {4},
  year      = {2006},
  month     = {Sep},
  publisher = {American Physical Society},
  doi       = {10.1103/PhysRevLett.97.120405},
  url       = {https://link.aps.org/doi/10.1103/PhysRevLett.97.120405}
}

@article{Acin2007,
  title     = {Device-Independent Security of Quantum Cryptography against Collective Attacks},
  author    = {Ac\'{\i}n, Antonio and Brunner, Nicolas and Gisin, Nicolas and Massar, Serge and Pironio, Stefano and Scarani, Valerio},
  journal   = {Phys. Rev. Lett.},
  volume    = {98},
  issue     = {23},
  pages     = {230501},
  numpages  = {4},
  year      = {2007},
  month     = {Jun},
  publisher = {American Physical Society},
  doi       = {10.1103/PhysRevLett.98.230501},
  url       = {https://link.aps.org/doi/10.1103/PhysRevLett.98.230501}
}

@article{Pironio2010,
  title     = {Random numbers certified by Bell’s theorem},
  volume    = {464},
  issn      = {1476-4687},
  url       = {http://dx.doi.org/10.1038/nature09008},
  doi       = {10.1038/nature09008},
  number    = {7291},
  journal   = {Nature},
  publisher = {Springer Science and Business Media LLC},
  author    = {Pironio, S. and Acín, A. and Massar, S. and de la Giroday, A. Boyer and Matsukevich, D. N. and Maunz, P. and Olmschenk, S. and Hayes, D. and Luo, L. and Manning, T. A. and Monroe, C.},
  year      = {2010},
  month     = Apr,
  pages     = {1021–1024}
}

@article{Cirelson1980,
  title        = {Quantum Generalizations of {{Bell}}'s Inequality},
  author       = {Cirel'son, B. S.},
  date         = {1980-03-01},
  journaltitle = {Letters in Mathematical Physics},
  journal      = {Lett. Math. Phys.},
  volume       = {4},
  number       = {2},
  pages        = {93--100},
  issn         = {1573-0530},
  year         = {1980},
  doi          = {10.1007/BF00417500},
  url          = {https://doi.org/10.1007/BF00417500}
}

@article{Aspect1982,
  title     = {Experimental Test of {Bell}'s Inequalities Using Time-Varying Analyzers},
  author    = {Aspect, Alain and Dalibard, Jean and Roger, G\'erard},
  journal   = {Phys. Rev. Lett.},
  volume    = {49},
  issue     = {25},
  pages     = {1804--1807},
  numpages  = {0},
  year      = {1982},
  month     = {Dec},
  publisher = {American Physical Society},
  doi       = {10.1103/PhysRevLett.49.1804},
  url       = {https://link.aps.org/doi/10.1103/PhysRevLett.49.1804}
}

@article{Hensen2015,
  author  = {Hensen, B. and Bernien, H. and Dréau, A. E. and Reiserer, A. and Kalb, N. and Blok, M. S. and Ruitenberg, J. and Vermeulen, R. F. L. and Schouten, R. N. and Abellán, C. and Amaya, W. and Pruneri, V. and Mitchell, M. W. and Markham, M. and Twitchen, D. J. and Elkouss, D. and Wehner, S. and Taminiau, T. H. and Hanson, R.},
  title   = {Loophole-free {B}ell inequality violation using electron spins separated by 1.3 kilometres},
  journal = {Nature},
  volume  = {526},
  pages   = {682--686},
  year    = {2015},
  issn    = {1476-4687},
  doi     = {10.1038/nature15759},
  url     = {https://doi.org/10.1038/nature15759}
}

@article{Giustina2015,
  title     = {Significant-Loophole-Free Test of {Bell's} Theorem with Entangled Photons},
  author    = {Giustina, Marissa and Versteegh, Marijn A. M. and Wengerowsky, S\"oren and Handsteiner, Johannes and Hochrainer, Armin and Phelan, Kevin and Steinlechner, Fabian and Kofler, Johannes and Larsson, Jan-\AA{}ke and Abell\'an, Carlos and Amaya, Waldimar and Pruneri, Valerio and Mitchell, Morgan W. and Beyer, J\"orn and Gerrits, Thomas and Lita, Adriana E. and Shalm, Lynden K. and Nam, Sae Woo and Scheidl, Thomas and Ursin, Rupert and Wittmann, Bernhard and Zeilinger, Anton},
  journal   = {Phys. Rev. Lett.},
  volume    = {115},
  issue     = {25},
  pages     = {250401},
  numpages  = {7},
  year      = {2015},
  month     = {Dec},
  publisher = {American Physical Society},
  doi       = {10.1103/PhysRevLett.115.250401},
  url       = {https://link.aps.org/doi/10.1103/PhysRevLett.115.250401}
}

@article{Shalm2015,
  title     = {Strong Loophole-Free Test of Local Realism},
  author    = {Shalm, Lynden K. and Meyer-Scott, Evan and Christensen, Bradley G. and Bierhorst, Peter and Wayne, Michael A. and Stevens, Martin J. and Gerrits, Thomas and Glancy, Scott and Hamel, Deny R. and Allman, Michael S. and Coakley, Kevin J. and Dyer, Shellee D. and Hodge, Carson and Lita, Adriana E. and Verma, Varun B. and Lambrocco, Camilla and Tortorici, Edward and Migdall, Alan L. and Zhang, Yanbao and Kumor, Daniel R. and Farr, William H. and Marsili, Francesco and Shaw, Matthew D. and Stern, Jeffrey A. and Abell\'an, Carlos and Amaya, Waldimar and Pruneri, Valerio and Jennewein, Thomas and Mitchell, Morgan W. and Kwiat, Paul G. and Bienfang, Joshua C. and Mirin, Richard P. and Knill, Emanuel and Nam, Sae Woo},
  journal   = {Phys. Rev. Lett.},
  volume    = {115},
  issue     = {25},
  pages     = {250402},
  numpages  = {10},
  year      = {2015},
  month     = {Dec},
  publisher = {American Physical Society},
  doi       = {10.1103/PhysRevLett.115.250402},
  url       = {https://link.aps.org/doi/10.1103/PhysRevLett.115.250402}
}

@article{Mermin1990,
  title     = {Extreme quantum entanglement in a superposition of macroscopically distinct states},
  author    = {Mermin, N. David},
  journal   = {Phys. Rev. Lett.},
  volume    = {65},
  issue     = {15},
  pages     = {1838--1840},
  numpages  = {0},
  year      = {1990},
  month     = {Oct},
  publisher = {American Physical Society},
  doi       = {10.1103/PhysRevLett.65.1838},
  url       = {https://link.aps.org/doi/10.1103/PhysRevLett.65.1838}
}

@incollection{GHZ1989,
  title     = {Going {{Beyond Bell}}'s {{Theorem}}},
  booktitle = {Bell's {{Theorem}}, {{Quantum Theory}} and {{Conceptions}} of the {{Universe}}},
  author    = {Greenberger, Daniel M. and Horne, Michael A. and Zeilinger, Anton},
  editor    = {Kafatos, Menas},
  year      = {1989},
  pages     = {69--72},
  publisher = {{Springer Netherlands}},
  address   = {{Dordrecht}},
  doi       = {10.1007/978-94-017-0849-4_10},
  url       = {https://doi.org/10.1007/978-94-017-0849-4_10},
  isbn      = {978-94-017-0849-4}
}

@article{Ardehali1992,
  title     = {{Bell} inequalities with a magnitude of violation that grows exponentially with the number of particles},
  author    = {Ardehali, M.},
  journal   = {Phys. Rev. A},
  volume    = {46},
  issue     = {9},
  pages     = {5375--5378},
  numpages  = {0},
  year      = {1992},
  month     = {Nov},
  publisher = {American Physical Society},
  doi       = {10.1103/PhysRevA.46.5375},
  url       = {https://link.aps.org/doi/10.1103/PhysRevA.46.5375}
}

@article{BelinskiiKlyshko1993,
  doi       = {10.1070/PU1993v036n08ABEH002299},
  url       = {https://doi.org/10.1070/PU1993v036n08ABEH002299},
  year      = {1993},
  month     = {aug},
  publisher = {},
  volume    = {36},
  number    = {8},
  pages     = {653},
  author    = {A V Belinskiĭ and D N Klyshko},
  title     = {Interference of light and {Bell's} theorem},
  journal   = {Phys. Usp.}
}

@article{ScaraniGisin2001,
  doi       = {10.1088/0305-4470/34/30/314},
  url       = {https://doi.org/10.1088/0305-4470/34/30/314},
  year      = {2001},
  month     = {jul},
  publisher = {},
  volume    = {34},
  number    = {30},
  pages     = {6043},
  author    = {Valerio Scarani and Nicolas Gisin},
  title     = {Spectral decomposition of {Bell's}operators for qubits},
  journal   = {J. Phys. A}
}

@article{Fan2023,
  title     = {Generalized iterative formula for Bell inequalities},
  author    = {Fan, XingYan and Xu, ZhenPeng and Miao, JiaLe and Liu, HongYe and Liu, YiJia and Shang, WeiMin and Zhou, Jie and Meng, HuiXian and G\"uhne, Otfried and Chen, JingLing},
  journal   = {Phys. Rev. A},
  volume    = {108},
  issue     = {6},
  pages     = {062404},
  numpages  = {33},
  year      = {2023},
  month     = {Dec},
  publisher = {American Physical Society},
  doi       = {10.1103/PhysRevA.108.062404},
  url       = {https://link.aps.org/doi/10.1103/PhysRevA.108.062404}
}

@article{WernerWolf2001,
  author    = {Werner, Reinhard F. and Wolf, Michael M.},
  title     = {All-multipartite {Bell}-correlation inequalities for two dichotomic observables per site},
  journal   = {Phys. Rev. A},
  volume    = {64},
  number    = {3},
  pages     = {032112},
  year      = {2001},
  doi       = {10.1103/PhysRevA.64.032112},
  publisher = {American Physical Society}
}

@article{ZukowskiBrukner2002,
  author    = {{\.{Z}}ukowski, Marek and Brukner, {\v C}aslav},
  title     = {{Bell}'s theorem for general {$N$}-qubit states},
  journal   = {Phys. Rev. Lett.},
  volume    = {88},
  number    = {21},
  pages     = {210401},
  year      = {2002},
  doi       = {10.1103/PhysRevLett.88.210401},
  publisher = {American Physical Society}
}

@article{Horodecki1995,
  author  = {Horodecki, Ryszard and Horodecki, Pawe{\l} and Horodecki, Micha{\l}},
  title   = {Violating {Bell} inequality by mixed spin-{$1/2$} states: Necessary and sufficient condition},
  journal = {Phys. Lett. A},
  volume  = {200},
  number  = {5},
  pages   = {340--344},
  year    = {1995},
  doi     = {10.1016/0375-9601(95)00214-N}
}

@article{ChenWuKwekOh2004,
  author    = {Chen, JingLing and Wu, Chunfeng and Kwek, L. C. and Oh, C. H.},
  title     = {{Gisin}'s theorem for three qubits},
  journal   = {Phys. Rev. Lett.},
  volume    = {93},
  number    = {14},
  pages     = {140407},
  year      = {2004},
  doi       = {10.1103/PhysRevLett.93.140407},
  publisher = {American Physical Society}
}

@article{ChenAlbeverioFei2006,
  author    = {Chen, Kai and Albeverio, Sergio and Fei, ShaoMing},
  title     = {Two-setting {Bell} inequalities for many qubits},
  journal   = {Phys. Rev. A},
  volume    = {74},
  number    = {5},
  pages     = {050101},
  year      = {2006},
  doi       = {10.1103/PhysRevA.74.050101},
  publisher = {American Physical Society}
}

@article{YuEtAl2012,
  author    = {Yu, Sixia and Chen, Qing and Zhang, Chengjie and Lai, C. H. and Oh, C. H.},
  title     = {All entangled pure states violate a single {Bell}'s inequality},
  journal   = {Phys. Rev. Lett.},
  volume    = {109},
  number    = {12},
  pages     = {120402},
  year      = {2012},
  doi       = {10.1103/PhysRevLett.109.120402},
  publisher = {American Physical Society}
}

@article{Kaniewski2016,
  author    = {Kaniewski, J{\k e}drzej},
  title     = {Analytic and nearly optimal self-testing bounds for the {Clauser--Horne--Shimony--Holt} and {Mermin} inequalities},
  journal   = {Phys. Rev. Lett.},
  volume    = {117},
  number    = {7},
  pages     = {070402},
  year      = {2016},
  doi       = {10.1103/PhysRevLett.117.070402},
  publisher = {American Physical Society}
}

@article{Panwar2023,
  author  = {Panwar, Ekta and Pandya, Palash and Wie{\'s}niak, Marcin},
  title   = {An elegant scheme of self-testing for multipartite {Bell} inequalities},
  journal = {npj Quantum Information},
  volume  = {9},
  pages   = {71},
  year    = {2023},
  doi     = {10.1038/s41534-023-00735-3}
}

\clearpage

\appendix

\makeatletter

\@addtoreset{equation}{section}
\renewcommand{\theequation}{\thesection.\arabic{equation}}

\@addtoreset{theorem}{section}
\renewcommand{\thetheorem}{\thesection.\arabic{theorem}}

\@addtoreset{lemma}{section}
\renewcommand{\thelemma}{\thesection.\arabic{lemma}}

\@addtoreset{definition}{section}
\renewcommand{\thedefinition}{\thesection.\arabic{definition}}

\@addtoreset{remark}{section}
\renewcommand{\theremark}{\thesection.\arabic{remark}}

\makeatother

\section{Von Neumann's Trace Inequality and Its Extension to Rectangular Matrices}\label{app:von Neumann trace inequality}

In this appendix, we establish the extension of the von Neumann's trace inequality valid for real rectangular matrices, which is central to the upper bound derivation in the main text.

\begin{theorem}[von Neumann trace inequality for rectangular matrices]
  \label{thm:von_neumann_rect}
  Let $X,Y\in\Real^{m\times n}$, set $r=\min(m,n)$, and let their singular values be sorted in non-increasing order
  \begin{align}
    \sigma_1(X)\ge\sigma_2(X)\ge\cdots\ge\sigma_r(X)\ge0, \\
    \sigma_1(Y)\ge\sigma_2(Y)\ge\cdots\ge\sigma_r(Y)\ge0.
  \end{align}
  Then the Frobenius inner product satisfies
  \begin{equation}
    \Tr(X^{\mathsf{T}}Y)\le\sum_{i=1}^{r}\sigma_i(X)\,\sigma_i(Y).
    \label{eq:main_ineq}
  \end{equation}
  Equality holds if and only if there exist matrices $U\in\Real^{m\times r}, V\in\Real^{n\times r}$ with orthonormal columns ($U^{\mathsf{T}}U=I_r$, $V^{\mathsf{T}}V=I_r$) such that
  \begin{equation}
    X=U\Sigma_X V^{\mathsf{T}},\quad Y=U\Sigma_Y V^{\mathsf{T}},
  \end{equation}
  where $\Sigma_X=\diag(\sigma_1(X),\dots,\sigma_r(X))$, $\Sigma_Y=\diag(\sigma_1(Y),\dots,\sigma_r(Y))$.
\end{theorem}

The proof of this theorem relies on the following lemma concerning the trace inequality for real symmetric matrices.

\begin{lemma}[Trace inequality for real symmetric matrices]
  \label{lem:sym_trace}
  Let $P,Q\in\Real^{N\times N}$ be real symmetric matrices, with eigenvalues sorted in descending order
  \[
    \lambda_1(P)\ge\cdots\ge\lambda_N(P),\qquad \lambda_1(Q)\ge\cdots\ge\lambda_N(Q).
  \]
  Then
  \begin{equation}
    \Tr(PQ)\le\sum_{i=1}^{N}\lambda_i(P)\lambda_i(Q).
    \label{eq:lemma_ineq}
  \end{equation}
  Equality holds if and only if there exists an orthogonal matrix $U$ such that
  \begin{align}
    P=U\diag(\lambda_1(P),\dots,\lambda_N(P))U^{\mathsf{T}}, \\
    Q=U\diag(\lambda_1(Q),\dots,\lambda_N(Q))U^{\mathsf{T}}.
  \end{align}
\end{lemma}

\begin{proof}
  \textbf{Step 1: Translation to positive semidefinite matrices.}
  Choose constants $c=-\lambda_N(P)$, $d=-\lambda_N(Q)$, so that $P+cI$ and $Q+dI$ are positive semidefinite.
  Expanding the trace product,
  \begin{equation}
    \begin{aligned}
        & \Tr\bigl((P+cI)(Q+dI)\bigr)  \\
      = & \Tr(PQ)+c\Tr(Q)+d\Tr(P)+Ncd,
    \end{aligned}
  \end{equation}
  \begin{equation}
    \begin{aligned}
        & \sum_{i=1}^{N}\bigl(\lambda_i(P)+c\bigr)\bigl(\lambda_i(Q)+d\bigr)                                  \\
      = & \sum_{i=1}^{N}\lambda_i(P)\lambda_i(Q)+c\sum_{i=1}^{N}\lambda_i(Q)+d\sum_{i=1}^{N}\lambda_i(P)+Ncd.
    \end{aligned}
  \end{equation}
  Since $\Tr(Q)=\sum\lambda_i(Q)$ and $\Tr(P)=\sum\lambda_i(P)$, the inequality \EquRef{eq:lemma_ineq} is equivalent to its validity for the positive semidefinite matrices $P+cI$, $Q+dI$, and the equality condition is invariant under translation.
  Hence we may assume without loss of generality that $P, Q$ are already positive semidefinite, so all their eigenvalues are nonnegative.

  \textbf{Step 2: Proof for the positive semidefinite case.}
  Take the spectral decomposition $P=U\Lambda U^{\mathsf{T}}$, where $\Lambda=\diag(\lambda_1,\dots,\lambda_N)$, $\lambda_i:=\lambda_i(P)\ge0$, and $U$ is orthogonal.
  Then
  \begin{align}
    \Tr(PQ)   & =\Tr(\Lambda\,U^{\mathsf{T}}QU)=\sum_{i=1}^{N}\lambda_i z_i,\nonumber \\
    \quad z_i & :=(U^{\mathsf{T}}QU)_{ii}\ge0.
  \end{align}
  Let $Q=WMW^{\mathsf{T}}$ be the spectral decomposition of $Q$, $M=\diag(\mu_1,\dots,\mu_N)$, $\mu_i:=\lambda_i(Q)\ge0$, with $W$ orthogonal.
  Define $V=U^{\mathsf{T}}W$,  which remains orthogonal. Then $z_i=\sum_{j=1}^{N}v_{ij}^2\mu_j$.
  Introduce $S_{ij}=v_{ij}^2$, the matrix $S$ is doubly stochastic (all entries nonnegative, all row and column sums equal to $1$).

  Set $z=(z_1,\dots,z_N)^{\mathsf{T}}$, so that $z=S\mu$, where $\mu=(\mu_1,\dots,\mu_N)^{\mathsf{T}}$.
  Let $z_1^{\downarrow}\ge\cdots\ge z_N^{\downarrow}$ be the components of $z$ arranged in descending order.
  By the property of doubly stochastic matrices, the vector $z$ is majorized by $\mu$:
  \begin{align}
     & \sum_{i=1}^{k} z_i^{\downarrow}\le\sum_{i=1}^{k}\mu_i\quad(k=1,\dots,N-1),\quad\nonumber \\
     & \sum_{i=1}^{N} z_i = \sum_{i=1}^{N}\mu_i.
    \label{eq:control}
  \end{align}
  For any index set $I\subseteq\{1,\dots,N\}$ with $|I|=k$,
  \[
    \sum_{i\in I}z_i = \sum_{j=1}^{N}\mu_j\Bigl(\sum_{i\in I}S_{ij}\Bigr).
  \]
  Let $t_j=\sum_{i\in I}S_{ij}$, then $0\le t_j\le1$ and $\sum_{j=1}^{N}t_j=k$. Under these constraints, $\sum\mu_j t_j$ is maximized by assigning as much weight as possible to the largest $\mu_j$, i.e., $\le\sum_{j=1}^{k}\mu_j$. Choosing $I$ to be the indices of the $k$ largest components of $z$ yields~\EquRef{eq:control}.

  Since $\lambda=(\lambda_1,\dots,\lambda_N)$ is a nonnegative sequence in descending order. By the rearrangement inequality,
  \[
    \sum_{i=1}^{N}\lambda_i z_i\le\sum_{i=1}^{N}\lambda_i z_i^{\downarrow}.
  \]
  Using the majorization relations and Abel summation:

  \begin{align}\label{eq:abel}
    \sum_{i=1}^{N}\lambda_i\mu_i & - \sum_{i=1}^{N}\lambda_i z_i^{\downarrow}
    = \sum_{k=1}^{N-1}(\lambda_k-\lambda_{k+1})\Bigl(\sum_{i=1}^{k}\mu_i-\sum_{i=1}^{k}z_i^{\downarrow}\Bigr) \nonumber \\
                                 & +\lambda_N\Bigl(\sum_{i=1}^{N}\mu_i-\sum_{i=1}^{N}z_i^{\downarrow}\Bigr) \ge0.
  \end{align}
  Here $\lambda_k-\lambda_{k+1}\ge0$, $\lambda_N\ge0$, and each bracketed expression is nonnegative by~\EquRef{eq:control}.
  Therefore,
  \begin{equation}
    \Tr(PQ)=\sum_{i=1}^{N}\lambda_i z_i \le \sum_{i=1}^{N}\lambda_i\mu_i,
  \end{equation}
  which is exactly \EquRef{eq:lemma_ineq}.

  \textbf{Step 3: Equality condition.}
  If equality holds, all terms in~\EquRef{eq:abel} must be zero. Because we can translate to make all $\lambda_i>0$ (choose $c>-\lambda_N(P)$), and translation does not affect the simultaneous diagonalization conclusion, we may assume without loss of generality that $\lambda_i>0$ for all $i$.
  Then $\lambda_k-\lambda_{k+1}>0$
  For brevity, assume $\lambda_i$ are strictly decreasing; then~\EquRef{eq:control} must hold with equality everywhere, i.e., for all $k$,
  $\sum_{i=1}^{k}z_i^{\downarrow}=\sum_{i=1}^{k}\mu_i$, which yields $z_i^{\downarrow}=\mu_i$.
  Since the rearrangement inequality attains equality (with $\lambda_i$ strictly decreasing), we must have $z_i=\mu_i$ and the $z_i$ are already in descending order.

  Now $Z=U^{\mathsf{T}}QU$ is a symmetric positive semidefinite matrix with diagonal entries $\mu_1,\dots,\mu_N$ and trace $\sum\mu_i$.
  But the eigenvalues of $Z$ are also $\mu_1,\dots,\mu_N$, so
  \begin{equation}
    \Tr(Z^2)=\sum_{i,j}Z_{ij}^2 = \sum_{i=1}^{N}\mu_i^2 + \sum_{i\neq j}Z_{ij}^2.
  \end{equation}
  On the other hand, $\Tr(Z^2)=\sum_{i=1}^{N}\mu_i^2$, hence all off-diagonal entries $Z_{ij}=0$ ($i\neq j$). Thus $Z=\diag(\mu_1,\dots,\mu_N)$, i.e., $U^{\mathsf{T}}QU=\diag(\mu_1,\dots,\mu_N)$. Meanwhile, $U^{\mathsf{T}}PU=\Lambda=\diag(\lambda_1,\dots,\lambda_N)$ holds by construction. Therefore $P$ and $Q$ are simultaneously diagonalized by the same orthogonal matrix, and the eigenvalues are arranged in descending order.
  If there were repeated eigenvalues, the conclusion remains that there exists an orthogonal matrix simultaneously diagonalizing them. Translating back to the original matrices does not destroy this property. The lemma is proved.
\end{proof}

With the lemma established, we now return to the proof of \ThmRef{thm:von_neumann_rect}.

\begin{proof}[Proof of Theorem \ref{thm:von_neumann_rect}]

  Let $N=\max(m,n)$. Pad $X$ and $Y$ with zeros to form square matrices:
  \begin{align}
    \widetilde{X}=
    \begin{cases}
      \begin{pmatrix} X \\ 0_{(n-m)\times n} \end{pmatrix},           & m\le n, \\
      \begin{pmatrix} X & 0_{m\times(m-n)} \end{pmatrix}, & m>n,
    \end{cases} \\
    \widetilde{Y}=
    \begin{cases}
      \begin{pmatrix} Y \\ 0_{(n-m)\times n} \end{pmatrix},           & m\le n, \\
      \begin{pmatrix} Y & 0_{m\times(m-n)} \end{pmatrix}, & m>n.
    \end{cases}
    \label{eq:fill}
  \end{align}
  In either case, $\widetilde{X},\widetilde{Y}\in\Real^{N\times N}$, and
  \begin{equation}
    \Tr(\widetilde{X}^{\mathsf{T}}\widetilde{Y})=\Tr(X^{\mathsf{T}}Y).
    \label{eq:trace_eq}
  \end{equation}
  Moreover, the singular values of $\widetilde{X}$ are $\sigma_1(X),\dots,\sigma_r(X)$ together with $N-r$ zeros, and similarly for $\widetilde{Y}$.
  If the trace inequality for square matrices holds, then
  \begin{equation}
    \begin{aligned}
      \Tr(X^{\mathsf{T}}Y) & =\Tr(\widetilde{X}^{\mathsf{T}}\widetilde{Y})                   \\
                           & \le\sum_{i=1}^{N}\sigma_i(\widetilde{X})\sigma_i(\widetilde{Y})
      =\sum_{i=1}^{r}\sigma_i(X)\sigma_i(Y).
    \end{aligned}
  \end{equation}
  Therefore, it suffices to prove the square matrix version for $m=n$, and to verify that the equality condition corresponds accordingly.
  Let $A,B\in\Real^{n\times n}$ have singular values $\sigma_1\ge\cdots\ge\sigma_n\ge0$ and $\tau_1\ge\cdots\ge\tau_n\ge0$, respectively.
  Construct the $2n\times 2n$ real symmetric matrices
  \begin{equation}
    \mathcal{A}= \begin{pmatrix} 0 & A \\ A^{\mathsf{T}} & 0 \end{pmatrix},\quad
    \mathcal{B}= \begin{pmatrix} 0 & B \\ B^{\mathsf{T}} & 0 \end{pmatrix}.
    \label{eq:sym_mat}
  \end{equation}
  A direct computation gives
  \begin{equation}
    \mathcal{A}\mathcal{B}= \begin{pmatrix} AB^{\mathsf{T}} & 0 \\ 0 & A^{\mathsf{T}}B \end{pmatrix},
  \end{equation}
  hence
  \begin{equation}
    \Tr(\mathcal{A}\mathcal{B}) = \Tr(AB^{\mathsf{T}})+\Tr(A^{\mathsf{T}}B)=2\Tr(A^{\mathsf{T}}B).
    \label{eq:trace_sym_prod}
  \end{equation}

  To obtain the eigenvalues of $\mathcal{A}$, take a full singular value decomposition of $A$ as $A=U\Sigma V^{\mathsf{T}}$, where $U, V$ are orthogonal and $\Sigma=\diag(\sigma_1,\dots,\sigma_n)$.
  Then
  \[
    \mathcal{A}= \begin{pmatrix} U & 0 \\ 0 & V \end{pmatrix}
    \begin{pmatrix} 0 & \Sigma \\ \Sigma & 0 \end{pmatrix}
    \begin{pmatrix} U^{\mathsf{T}} & 0 \\ 0 & V^{\mathsf{T}} \end{pmatrix}.
  \]
  The matrix $\left(\begin{smallmatrix}0&\Sigma\\ \Sigma&0\end{smallmatrix}\right)$ is orthogonally similar to $\bigoplus_{i=1}^{n}\left(\begin{smallmatrix}0&\sigma_i\\ \sigma_i&0\end{smallmatrix}\right)$,
  and each $2\times2$ block has eigenvalues $\pm\sigma_i$. Hence the eigenvalues of $\mathcal{A}$ (in descending order) are
  \begin{equation}
    \lambda(\mathcal{A}) : \sigma_1,\sigma_2,\dots,\sigma_n,-\sigma_n,-\sigma_{n-1},\dots,-\sigma_1.
    \label{eq:eig_A}
  \end{equation}
  Similarly, the eigenvalues of $\mathcal{B}$ in descending order are
  \begin{equation}
    \lambda(\mathcal{B}) : \tau_1,\tau_2,\dots,\tau_n,-\tau_n,-\tau_{n-1},\dots,-\tau_1.
    \label{eq:eig_B}
  \end{equation}

  We now apply \LemRef{lem:sym_trace} to $\mathcal{A}$ and $\mathcal{B}$ (here $N=2n$):

  \begin{equation}
    \begin{aligned}
      \Tr(\mathcal{A}\mathcal{B})
       & \le\sum_{i=1}^{2n}\lambda_i(\mathcal{A})\lambda_i(\mathcal{B})             \\
       & =\sum_{i=1}^{n}\bigl(\sigma_i\tau_i+(-\sigma_{n-i+1})(-\tau_{n-i+1})\bigr) \\
       & =2\sum_{i=1}^{n}\sigma_i\tau_i.
    \end{aligned}
  \end{equation}

  Combining with \EquRef{eq:trace_sym_prod} immediately yields the inequality for square matrices
  \begin{equation}
    \Tr(A^{\mathsf{T}}B)\le\sum_{i=1}^{n}\sigma_i(A)\sigma_i(B).
    \label{eq:square_ineq}
  \end{equation}

  In the square matrix case, by the equality condition of \LemRef{lem:sym_trace}, $\Tr(A^{\mathsf{T}}B)=\sum_{i=1}^{n}\sigma_i\tau_i$ holds if and only if there exists an orthogonal matrix $\mathcal{U}$ that simultaneously diagonalizes $\mathcal{A}$ and $\mathcal{B}$. This is equivalent to the existence of orthogonal matrices $U,V\in\Real^{n\times n}$ such that
  \begin{equation}
    A=U\Sigma_A V^{\mathsf{T}},\qquad B=U\Sigma_B V^{\mathsf{T}},
  \end{equation}
  where $\Sigma_A=\diag(\sigma_1,\dots,\sigma_n)$ and $\Sigma_B=\diag(\tau_1,\dots,\tau_n)$. That is, $A$ and $B$ share the same left and right singular vectors, with corresponding singular values in the same order.

  Now return to the padded square matrices $\widetilde{X},\widetilde{Y}\in\Real^{N\times N}$. If equality holds in the rectangular inequality \EquRef{eq:main_ineq}, then equality holds for the square matrix case as well. Hence there exist orthogonal matrices $\widetilde{U},\widetilde{V}\in\Real^{N\times N}$ such that
  \begin{equation}
    \widetilde{X}=\widetilde{U}\widetilde{\Sigma}_X\widetilde{V}^{\mathsf{T}},\quad
    \widetilde{Y}=\widetilde{U}\widetilde{\Sigma}_Y\widetilde{V}^{\mathsf{T}},
  \end{equation}
  where $\widetilde{\Sigma}_X=\diag(\sigma_1(X),\dots,\sigma_r(X),0,\dots,0)$, and similarly for $\widetilde{\Sigma}_Y$.
  According to the padding construction (see \EquRef{eq:fill}):
  \begin{enumerate}
    \item[1)] If $m\le n$ ($N=n$), the last $n-m$ rows of $\widetilde{X}$ are zero, so the last $n-m$ rows of the first $r$ columns of $\widetilde{U}$ (corresponding to the nonzero singular values) are zero. Taking the first $m$ rows of these columns gives a matrix $U\in\Real^{m\times r}$ with $U^{\mathsf{T}}U=I_r$, the first $r$ columns of $\widetilde{V}$ are denoted $V\in\Real^{n\times r}$ and satisfy $V^{\mathsf{T}}V=I_r$.
    \item[2)] If $m>n$ ($N=m$), symmetrically, the last $m-n$ rows of the first $r$ columns of $\widetilde{V}$ (corresponding to the nonzero singular values) are zero. Taking the first $n$ rows yields $V\in\Real^{n\times r}$ with $V^{\mathsf{T}}V=I_r$, the first $r$ columns of $\widetilde{U}$ give $U\in\Real^{m\times r}$ with $U^{\mathsf{T}}U=I_r$.
  \end{enumerate}

  In either case we obtain
  \begin{equation}
    X=U\Sigma_X V^{\mathsf{T}},\qquad Y=U\Sigma_Y V^{\mathsf{T}},
  \end{equation}
  where $\Sigma_X=\diag(\sigma_1(X),\dots,\sigma_r(X))$ and $\Sigma_Y=\diag(\sigma_1(Y),\dots,\sigma_r(Y))$.
  This is precisely the necessary and sufficient condition for equality to hold in the rectangular case.

  This completes the proof of \ThmRef{thm:von_neumann_rect}.

\end{proof}

\section{The Norm of MABK and EMABK Bell Polynomials is Constant}\label{app:mabk-norm}

This appendix proves that the Frobenius norm of the effective-coefficient tensor for both the standard MABK and the EMABK constructions is constant and independent of the number of parties, the specific measurement directions, and the choice of the constituent MABK blocks.

\begin{theorem}
  Under the iterative formula
  \begin{align}
    \mathcal{B}_2 & =\frac{1}{2} \left(M^{(1)}_1 M^{(2)}_1 + M^{(1)}_1 M^{(2)}_2 + M^{(1)}_2 M^{(2)}_1 - M^{(1)}_2 M^{(2)} _2\right),                                                         \\
    \mathcal{B}_n & =\frac{1}{2} \left[  \mathcal{B}_{n-1}^{(++)}\left(M^{(n)}_1 + M^{(n)}_2\right) + \mathcal{B}_{n-1}^{(+-)}\left(M^{(n)}_1 - M^{(n)}_2\right) \right]\label{eq:MABK-itera}
  \end{align}
  the Frobenius norm of the MABK effective coefficient tensor satisfies
  \begin{equation}
    \|\mathbf{C}\|_F^2 =1
    \label{eq:MABK_C_norm}
  \end{equation}
  and is independent of the specific measurement directions. Here, $\mathcal{B}_{n-1}^{(+-)}$ denotes the operator obtained by interchanging all measurement settings $1\leftrightarrow 2$ in $\mathcal{B}_{n-1}^{(++)}$.

\end{theorem}

The proof of this theorem relies on the following lemmas and auxiliary results, which we establish first.

\begin{lemma}[Recursive Relation of the Effective Coefficient Tensor]
  \label{lem:spatial-recursion}
  Let the measurement operator of the $k$-th party be $M_i^{(k)} = \sum_{p=1}^3 a_{ip}^{(k)} \sigma_p^{(k)}$ ($i=1,2$),
  and define the sum and difference vectors
  \begin{equation}
    S^{(k)}_p = a_{1p}^{(k)} + a_{2p}^{(k)}, \quad
    D^{(k)}_p = a_{1p}^{(k)} - a_{2p}^{(k)}.
  \end{equation}
  Expanding $\mathcal{B}_n$ in the product basis of Pauli matrices:
  \begin{equation}
    \mathcal{B}_n = \sum_{p_1,\dots,p_n=1}^3 C^{(n)}_{p_1\dots p_n}\, \sigma_{p_1}^{(1)}\cdots\sigma_{p_n}^{(n)},
  \end{equation}
  where $C^{(n)}_{p_1\dots p_n}$ is the effective coefficient tensor of $\mathcal{B}_n^{(++)}$, and $\widetilde{C}^{(n)}$ is the effective coefficient tensor of $\mathcal{B}_n^{(+-)}$ (obtained by interchanging all $M_1\leftrightarrow M_2$).
  Then for $n\ge 3$, the tensors satisfy the recursion:
  \begin{equation}\label{eq:spatial-recursion}
    \begin{aligned}
      C^{(n)}_{p_1\dots p_n}             & = \frac{1}{2}\Bigl(S^{(n)}_{p_n}\,C^{(n-1)}_{p_1\dots p_{n-1}} + D^{(n)}_{p_n}\,\widetilde{C}^{(n-1)}_{p_1\dots p_{n-1}}\Bigr), \\[4pt]
      \widetilde{C}^{(n)}_{p_1\dots p_n} & = \frac{1}{2}\Bigl(S^{(n)}_{p_n}\,\widetilde{C}^{(n-1)}_{p_1\dots p_{n-1}} - D^{(n)}_{p_n}\,C^{(n-1)}_{p_1\dots p_{n-1}}\Bigr).
    \end{aligned}
  \end{equation}
  The initial conditions for $n=2$ are
  \begin{equation}
    \label{eq:spatial-init}
    \begin{aligned}
      C^{(2)}_{p_1p_2}             & = \frac{1}{2}\bigl(a^{(1)}_{1p_1}S^{(2)}_{p_2} + a^{(1)}_{2p_1}D^{(2)}_{p_2}\bigr), \\[2pt]
      \widetilde{C}^{(2)}_{p_1p_2} & = \frac{1}{2}\bigl(a^{(1)}_{2p_1}S^{(2)}_{p_2} - a^{(1)}_{1p_1}D^{(2)}_{p_2}\bigr).
    \end{aligned}
  \end{equation}
\end{lemma}

\begin{proof}
  We first address the recursive relation for $n\ge 3$. From the recursive construction of MABK, we have
  \begin{equation}
    \label{eq:MABK-recursion-proof}
    \mathcal{B}_n = \frac{1}{2}\Bigl[\mathcal{B}_{n-1}^{(++)}\otimes\bigl(M_1^{(n)}+M_2^{(n)}\bigr) + \mathcal{B}_{n-1}^{(+-)}\otimes\bigl(M_1^{(n)}-M_2^{(n)}\bigr)\Bigr],
  \end{equation}
  where $\mathcal{B}_{n-1}^{(++)} = \mathcal{B}_{n-1}$ is the standard $(n-1)$-party MABK operator, and $\mathcal{B}_{n-1}^{(+-)} = \widetilde{\mathcal{B}}_{n-1}$ is the operator obtained by interchanging all measurement settings $1\leftrightarrow 2$.
  Expanding $\mathcal{B}_{n-1}$ and $\widetilde{\mathcal{B}}_{n-1}$ in terms of the effective coefficient tensors:
  \begin{align}
    \mathcal{B}_{n-1}             & = \sum_{p_1,\dots,p_{n-1}} C^{(n-1)}_{p_1\dots p_{n-1}}\, \sigma_{p_1}^{(1)}\cdots\sigma_{p_{n-1}}^{(n-1)}, \label{eq:B-expand}                   \\
    \widetilde{\mathcal{B}}_{n-1} & = \sum_{p_1,\dots,p_{n-1}} \widetilde{C}^{(n-1)}_{p_1\dots p_{n-1}}\, \sigma_{p_1}^{(1)}\cdots\sigma_{p_{n-1}}^{(n-1)}. \label{eq:B-tilde-expand}
  \end{align}
  Meanwhile, using the definitions of the sum and difference vectors, the measurement combinations for the $n$-th party can be written as
  \begin{align}
    M_1^{(n)}+M_2^{(n)} & = \sum_{p_n=1}^3 \bigl(a_{1p_n}^{(n)}+a_{2p_n}^{(n)}\bigr)\sigma_{p_n}^{(n)} = \sum_{p_n} S^{(n)}_{p_n}\sigma_{p_n}^{(n)}, \label{eq:S-def} \\
    M_1^{(n)}-M_2^{(n)} & = \sum_{p_n=1}^3 \bigl(a_{1p_n}^{(n)}-a_{2p_n}^{(n)}\bigr)\sigma_{p_n}^{(n)} = \sum_{p_n} D^{(n)}_{p_n}\sigma_{p_n}^{(n)}. \label{eq:D-def}
  \end{align}

  Substituting \EquRef{eq:B-expand}--\EquRef{eq:D-def} into \EquRef{eq:MABK-recursion-proof}:
  \begin{align*}
    \mathcal{B}_n = \frac{1}{2}\Biggl[
      & \Bigl(\sum_{P} C^{(n-1)}_P \sigma_P\Bigr) \otimes \Bigl(\sum_{p_n} S^{(n)}_{p_n} \sigma_{p_n}^{(n)}\Bigr)             \\
    + & \Bigl(\sum_{P} \widetilde{C}^{(n-1)}_P \sigma_P\Bigr) \otimes \Bigl(\sum_{p_n} D^{(n)}_{p_n} \sigma_{p_n}^{(n)}\Bigr)
      \Biggr],
  \end{align*}
  where $P=(p_1,\dots,p_{n-1})$ denotes the first $n-1$ spatial indices, and $\sigma_P = \sigma_{p_1}^{(1)}\cdots\sigma_{p_{n-1}}^{(n-1)}$.

  Expanding the two tensor products and combining the sums:
  \begin{align}
    \mathcal{B}_n & = \frac{1}{2} \sum_{P} \sum_{p_n} \Bigl( C^{(n-1)}_P S^{(n)}_{p_n} + \widetilde{C}^{(n-1)}_P D^{(n)}_{p_n} \Bigr) \sigma_P \otimes \sigma_{p_n}^{(n)} \nonumber                                 \\
                  & = \sum_{p_1,\dots,p_n} \frac{1}{2}\Bigl( C^{(n-1)}_{p_1\dots p_{n-1}} S^{(n)}_{p_n} + \widetilde{C}^{(n-1)}_{p_1\dots p_{n-1}} D^{(n)}_{p_n} \Bigr) \sigma_{p_1}^{(1)}\cdots\sigma_{p_n}^{(n)}.
    \label{eq:Bn-expanded}
  \end{align}
  Since $\{\sigma_{p_1}^{(1)}\cdots\sigma_{p_n}^{(n)}\}_{p_1,\dots,p_n=1}^3$ forms a linearly independent basis (different combinations of spatial indices provide a basis for the $n$-body operator space), the coefficient of $\sigma_{p_1}^{(1)}\cdots\sigma_{p_n}^{(n)}$ in \EquRef{eq:Bn-expanded} must equal $C^{(n)}_{p_1\dots p_n}$. Comparing term by term yields
  \begin{equation}
    C^{(n)}_{p_1\dots p_n} = \frac{1}{2}\Bigl( C^{(n-1)}_{p_1\dots p_{n-1}} S^{(n)}_{p_n} + \widetilde{C}^{(n-1)}_{p_1\dots p_{n-1}} D^{(n)}_{p_n} \Bigr).
  \end{equation}

  To derive the recursion for $\widetilde{C}^{(n)}$, consider interchanging all $M_1\leftrightarrow M_2$ in $\mathcal{B}_n$. The effect of this operation is as follows: (1) For the $n$-th party: $S^{(n)} = \vec a_1^{(n)}+\vec a_2^{(n)} \to \vec a_2^{(n)}+\vec a_1^{(n)} = S^{(n)}$ (unchanged), while $D^{(n)} = \vec a_1^{(n)}-\vec a_2^{(n)} \to \vec a_2^{(n)}-\vec a_1^{(n)} = -D^{(n)}$ (sign flip); (2) For the first $n-1$ parties: $C^{(n-1)} \leftrightarrow \widetilde{C}^{(n-1)}$.
  Therefore
  \begin{equation}
    \widetilde{C}^{(n)}_{p_1\dots p_n} = \frac{1}{2}\Bigl( \widetilde{C}^{(n-1)}_{p_1\dots p_{n-1}} S^{(n)}_{p_n} - C^{(n-1)}_{p_1\dots p_{n-1}} D^{(n)}_{p_n} \Bigr).
  \end{equation}
  This proves \EquRef{eq:spatial-recursion}.

  Next, we verify the initial condition for $n=2$. In the normalized version, the CHSH coefficients are $c_{11}=c_{12}=c_{21}=1/2$, $c_{22}=-1/2$. According to the definition of the effective coefficient tensor:
  \begin{align}
    C^{(2)}_{p_1p_2} & = \sum_{i_1,i_2=1}^2 c_{i_1i_2}\, a^{(1)}_{i_1p_1} a^{(2)}_{i_2p_2}                                                                                                       \nonumber \\
                     & = \frac12\,a_{1p_1}^{(1)}a_{1p_2}^{(2)} + \frac12\,a_{1p_1}^{(1)}a_{2p_2}^{(2)} + \frac12\,a_{2p_1}^{(1)}a_{1p_2}^{(2)} - \frac12\,a_{2p_1}^{(1)}a_{2p_2}^{(2)} \nonumber           \\
                     & = \frac12\Bigl( a_{1p_1}^{(1)}\bigl(a_{1p_2}^{(2)}+a_{2p_2}^{(2)}\bigr) + a_{2p_1}^{(1)}\bigl(a_{1p_2}^{(2)}-a_{2p_2}^{(2)}\bigr) \Bigr) \nonumber                                  \\
                     & = \frac12\bigl( a_{1p_1}^{(1)}S^{(2)}_{p_2} + a_{2p_1}^{(1)}D^{(2)}_{p_2} \bigr).
  \end{align}

  For $\widetilde{C}^{(2)}$, it is the expansion coefficient of the operator obtained by interchanging all $M_1\leftrightarrow M_2$ in $\mathcal{B}_2$. Equivalently, $\widetilde{C}^{(2)}$ corresponds to interchanging all $\vec a_1^{(k)}$ and $\vec a_2^{(k)}$ in $C^{(2)}$. For $k=1$, $a_{1p_1}^{(1)}\leftrightarrow a_{2p_1}^{(1)}$; for $k=2$, $S^{(2)}$ remains unchanged (since $\vec a_1^{(2)}+\vec a_2^{(2)}$ is symmetric), while $D^{(2)}\to -D^{(2)}$ (since $\vec a_1^{(2)}-\vec a_2^{(2)}$ is antisymmetric). Thus
  \begin{align}
    \widetilde{C}^{(2)}_{p_1p_2} & = \frac12\bigl( a_{2p_1}^{(1)}S^{(2)}_{p_2} + a_{1p_1}^{(1)}(-D^{(2)}_{p_2}) \bigr) \nonumber \\
                                 & = \frac12\bigl( a_{2p_1}^{(1)}S^{(2)}_{p_2} - a_{1p_1}^{(1)}D^{(2)}_{p_2} \bigr).
  \end{align}
  This verifies \EquRef{eq:spatial-init}.
\end{proof}

To disentangle the coupled recursion relations \EquRef{eq:spatial-recursion} into a single multiplicative iteration, we encode the two real tensors into one complex tensor as follows.

\begin{lemma}[Multiplicative Iteration of Complex Tensors]
  \label{lem:complex-iteration}
  Define the complex tensor
  \begin{equation}
    Z^{(n)}_{p_1\dots p_n} = C^{(n)}_{p_1\dots p_n} + i\,\widetilde{C}^{(n)}_{p_1\dots p_n},
  \end{equation}
  then the coupled system \EquRef{eq:spatial-recursion} can be written as
  \begin{equation}
    \label{eq:Z-iteration}
    Z^{(n)}_{p_1\dots p_n} = \frac{1}{2}\Bigl(S^{(n)}_{p_n} - i D^{(n)}_{p_n}\Bigr)\, Z^{(n-1)}_{p_1\dots p_{n-1}}, \qquad (n\ge 3)
  \end{equation}
  where
  \begin{equation}
    \begin{aligned}
      \frac{1}{2}\bigl(S^{(n)}-iD^{(n)}\bigr)
       & = \frac{1-i}{2}\,\vec a_1^{(n)} + \frac{1+i}{2}\,\vec a_2^{(n)}                       \\
       & = \frac{1}{\sqrt{2}}\Bigl(e^{-i\pi/4}\vec a_1^{(n)} + e^{i\pi/4}\vec a_2^{(n)}\Bigr).
    \end{aligned}
  \end{equation}
\end{lemma}

\begin{proof}
  Direct verification yields:
  \begin{align*}
    \frac{1}{2}(S-iD)\cdot(C+i\widetilde{C})
     & = \frac{1}{2}\left(SC + D\widetilde{C}\right) + \frac{i}{2}\left(S\widetilde{C} - DC\right) \\
     & = C^{(n)} + i\,\widetilde{C}^{(n)} = Z^{(n)},
  \end{align*}
  which is consistent with \EquRef{eq:spatial-recursion}. Moreover,
  \begin{align*}
    \frac{1}{2}(S-iD) & = \frac{1}{2}\bigl[(\vec a_1+\vec a_2) - i(\vec a_1-\vec a_2)\bigr] \\
                      & = \frac{1-i}{2}\,\vec a_1 + \frac{1+i}{2}\,\vec a_2,
  \end{align*}
  where $(1-i)/2 = e^{-i\pi/4}/\sqrt{2}$ and $(1+i)/2 =e^{i\pi/4}/\sqrt{2}$, which yields the desired form.
\end{proof}

Iterating the above multiplicative relation yields an explicit tensor-product representation of $Z^{(n)}$, which is summarized below.
\begin{lemma}[Effective Coefficient Tensor Expressed as a Complex Tensor]
  \label{lem:tensor-product}
  Define the complex vectors for each party:
  \begin{align}
    \label{eq:v-vectors}
    v^{(1)} & = \bigl(\vec a_1^{(1)} + i\,\vec a_2^{(1)}\bigr),                                         \\
    v^{(k)} & = \frac{1-i}{2}\,\vec a_1^{(k)} + \frac{1+i}{2}\,\vec a_2^{(k)}, \quad (k\ge 2).\nonumber
  \end{align}
  Then for all $n\ge 2$,
  \begin{equation}
    \label{eq:Z-tensor-product}
    Z^{(n)}_{p_1\dots p_n} = v^{(1)}_{p_1}\, v^{(2)}_{p_2}\cdots v^{(n)}_{p_n},
  \end{equation}
  that is, $Z^{(n)} = v^{(1)}\otimes v^{(2)}\otimes\cdots\otimes v^{(n)}$, and
  \begin{equation}
    \label{eq:C-from-Z}
    C^{(n)}_{p_1\dots p_n} = \operatorname{Re}\!\bigl[Z^{(n)}_{p_1\dots p_n}\bigr].
  \end{equation}
\end{lemma}

\begin{proof}
  By \LemRef{lem:complex-iteration}, for $n\ge 3$ we have $Z^{(n)}_{p_1\dots p_n} = v^{(n)}_{p_n} Z^{(n-1)}_{p_1\dots p_{n-1}}$.
  For $n=2$, from \EquRef{eq:spatial-init}:
  \begin{align*}
    Z^{(2)}_{p_1p_2} & = C^{(2)}_{p_1p_2} + i\,\widetilde{C}^{(2)}_{p_1p_2}                                                                                                                    \\
                     & = \frac{1}{2}\Bigl[\bigl(a_{1p_1}^{(1)}S^{(2)}_{p_2} + a_{2p_1}^{(1)}D^{(2)}_{p_2}\bigr) + i\bigl(a_{2p_1}^{(1)}S^{(2)}_{p_2} - a_{1p_1}^{(1)}D^{(2)}_{p_2}\bigr)\Bigr] \\
                     & = \frac{1}{2}\Bigl[ a_{1p_1}^{(1)}(S^{(2)}_{p_2}-iD^{(2)}_{p_2}) + i\,a_{2p_1}^{(1)}(S^{(2)}_{p_2}-iD^{(2)}_{p_2}) \Bigr]                                               \\
                     & = \frac{1}{2}(a_{1p_1}^{(1)}+i\,a_{2p_1}^{(1)})(S^{(2)}_{p_2}-iD^{(2)}_{p_2})                                                                                           \\
                     & = \frac{1}{2}v^{(1)}_{p_1}\cdot 2v^{(2)}_{p_2} = v^{(1)}_{p_1}v^{(2)}_{p_2}.
  \end{align*}
  Equation \EquRef{eq:Z-tensor-product} then follows by induction. Equation \EquRef{eq:C-from-Z} follows directly from the definition $C^{(n)} = \operatorname{Re}[Z^{(n)}]$.
\end{proof}

With the tensor-product structure established, we now turn to computing the Frobenius norm of the real coefficient tensor. The following lemma provides a decomposition that separates the norm into two tractable terms.

\begin{lemma}[Decomposition Identity for the Frobenius Norm]
  \label{lem:norm-decomposition}
  Let $Z^{(n)}_{p_1\dots p_n} = v^{(1)}_{p_1}\cdots v^{(n)}_{p_n}$ be the complex tensor defined in \LemRef{lem:tensor-product}, and $C^{(n)}_{p_1\dots p_n} = \operatorname{Re}[Z^{(n)}_{p_1\dots p_n}]$. Then the square of its Frobenius norm can be decomposed as
  \begin{equation}
    \label{eq:norm-decomposition-thm}
    \|\mathbf{C}^{(n)}\|_F^2 = \frac{1}{2}\sum_P|Z_P|^2 + \frac{1}{2}\operatorname{Re}\!\left(\sum_P Z_P^2\right),
  \end{equation}
  where $P=(p_1,\dots,p_n)$, and the sums can be decomposed into products of independent contributions from each party:
  \begin{align}
    \label{eq:mod-square-thm}
    \sum_P |Z_P|^2 & = \prod_{k=1}^n \|v^{(k)}\|^2,                    \\
    \label{eq:square-thm}
    \sum_P Z_P^2   & = \prod_{k=1}^n \bigl(v^{(k)}\cdot v^{(k)}\bigr).
  \end{align}
\end{lemma}

\begin{proof}
  From $C^{(n)}_{p_1\dots p_n} = \operatorname{Re}[Z^{(n)}_{p_1\dots p_n}]$, the square of the Frobenius norm is
  \begin{equation}
    \|\mathbf{C}^{(n)}\|_F^2 = \sum_{p_1,\dots,p_n=1}^3 \bigl(C^{(n)}_{p_1\dots p_n}\bigr)^2 = \sum_P \bigl(\operatorname{Re}\,Z_P\bigr)^2.
  \end{equation}
  For any complex number $z$, using the identity
  \begin{equation}
    (\operatorname{Re}\,z)^2 = \frac{z^2 + 2|z|^2 + \bar{z}^2}{4},
  \end{equation}
  we obtain
  \begin{align}
    \|\mathbf{C}^{(n)}\|_F^2
     & = \frac{1}{4}\sum_P\left(Z_P^2 + 2|Z_P|^2 + \bar{Z}_P^2\right) \nonumber              \\
     & = \frac{1}{2}\sum_P|Z_P|^2 + \frac{1}{2}\operatorname{Re}\!\left(\sum_P Z_P^2\right),
  \end{align}
  where we used $\sum_P \bar{Z}_P^2 = \overline{\sum_P Z_P^2}$, so $\sum_P Z_P^2 + \sum_P \bar{Z}_P^2 = 2\operatorname{Re}(\sum_P Z_P^2)$.

  Since $Z_P = v^{(1)}_{p_1}\cdots v^{(n)}_{p_n}$ is of product form, the sum over $P$ can be decomposed:
  \begin{equation}
    \begin{aligned}
      \sum_P |Z_P|^2
       & = \sum_{p_1,\dots,p_n} |v^{(1)}_{p_1}|^2\cdots|v^{(n)}_{p_n}|^2 \\
       & = \prod_{k=1}^n \left(\sum_{p=1}^3 |v^{(k)}_p|^2\right)
      = \prod_{k=1}^n \|v^{(k)}\|^2,                                     \\
    \end{aligned}
  \end{equation}
  \begin{equation}
    \begin{aligned}
      \sum_P Z_P^2
       & = \sum_{p_1,\dots,p_n} (v^{(1)}_{p_1})^2\cdots(v^{(n)}_{p_n})^2 \\
       & = \prod_{k=1}^n \left(\sum_{p=1}^3 (v^{(k)}_p)^2\right)
      = \prod_{k=1}^n \bigl(v^{(k)}\cdot v^{(k)}\bigr),
    \end{aligned}
  \end{equation}
  where $v^{(k)}\cdot v^{(k)} = \sum_p (v^{(k)}_p)^2$ is the inner product of the complex vector with itself (without conjugation).
\end{proof}

To evaluate the decomposed norm explicitly, we require the squared norms and self-inner products of the complex vectors $v^{(k)}$. These are given by the next lemma.

\begin{lemma}[Key Quantities of the Complex Vectors]
  \label{lem:vector-quantities}
  Let the measurement directions of the $k$-th party be unit vectors $\vec a_1^{(k)}, \vec a_2^{(k)}$ in $\mathbb{R}^3$, with the complex vectors $v^{(k)}$ defined as in \EquRef{eq:v-vectors}. Denote $\alpha = (1-i)/2$ and $\beta = (1+i)/2$. Then the squared norms and self-inner products of the complex vectors are given by:
  (1) First party,
  \begin{align}
    \label{eq:v1-norm-thm}
    \|v^{(1)}\|^2          & = \|\vec a_1^{(1)}\|^2 + \|\vec a_2^{(1)}\|^2,                                                       \\
    \label{eq:v1-dot-thm}
    v^{(1)}\!\cdot v^{(1)} & = \bigl(\|\vec a_1^{(1)}\|^2 - \|\vec a_2^{(1)}\|^2\bigr) + 2i\,\vec a_1^{(1)}\!\cdot\vec a_2^{(1)}.
  \end{align}
  $k$-th party $(k\ge 2)$,
  \begin{align}
    \label{eq:vk-norm-thm}
    \|v^{(k)}\|^2          & = \frac{1}{2}\bigl(\|\vec a_1^{(k)}\|^2 + \|\vec a_2^{(k)}\|^2\bigr),                                       \\
    \label{eq:vk-dot-thm}
    v^{(k)}\!\cdot v^{(k)} & = \vec a_1^{(k)}\!\cdot\vec a_2^{(k)} + \frac{i}{2}\bigl(\|\vec a_2^{(k)}\|^2 - \|\vec a_1^{(k)}\|^2\bigr).
  \end{align}
\end{lemma}

\begin{proof}
  From $|\alpha|^2 = |\beta|^2 = \frac{1}{2}$ and $\alpha\bar{\beta} = \frac{-i}{2}$, $\alpha^2 = \frac{-i}{2}$, $\beta^2 = \frac{i}{2}$, $2\alpha\beta = 1$.

  (2)First party, $v^{(1)} = \vec a_1^{(1)} + i\,\vec a_2^{(1)}$.
  Squared norm:
  \begin{align}
    \|v^{(1)}\|^2
     & = (\vec a_1^{(1)} + i\,\vec a_2^{(1)})\cdot(\vec a_1^{(1)} - i\,\vec a_2^{(1)}) \nonumber                                        \\
     & = \|\vec a_1^{(1)}\|^2 + \|\vec a_2^{(1)}\|^2 + 2\operatorname{Re}\!\bigl(i\,\vec a_1^{(1)}\!\cdot\vec a_2^{(1)}\bigr) \nonumber \\
     & = \|\vec a_1^{(1)}\|^2 + \|\vec a_2^{(1)}\|^2,
  \end{align}
  where the cross term is purely imaginary since $i$ multiplied by a real inner product is imaginary, so its real part vanishes.
  Self-inner product (without conjugation):
  \begin{align}
    v^{(1)}\!\cdot v^{(1)}
     & = (\vec a_1^{(1)} + i\,\vec a_2^{(1)})^2 \nonumber                                                   \\
     & = \|\vec a_1^{(1)}\|^2 + 2i\,\vec a_1^{(1)}\!\cdot\vec a_2^{(1)} + i^2\|\vec a_2^{(1)}\|^2 \nonumber \\
     & = \bigl(\|\vec a_1^{(1)}\|^2 - \|\vec a_2^{(1)}\|^2\bigr) + 2i\,\vec a_1^{(1)}\!\cdot\vec a_2^{(1)}.
  \end{align}
  $k$-th party $(k\ge 2)$, $v^{(k)} = \alpha\,\vec a_1^{(k)} + \beta\,\vec a_2^{(k)}$.
  Squared norm:
  \begin{align}
    \|v^{(k)}\|^2
     & = |\alpha|^2\|\vec a_1^{(k)}\|^2 + |\beta|^2\|\vec a_2^{(k)}\|^2 + 2\operatorname{Re}\!\bigl(\alpha\bar{\beta}\,\vec a_1^{(k)}\!\cdot\vec a_2^{(k)}\bigr) \nonumber \\
     & = \frac{1}{2}\|\vec a_1^{(k)}\|^2 + \frac{1}{2}\|\vec a_2^{(k)}\|^2 + 2\operatorname{Re}\!\Bigl(\frac{-i}{2}\,\vec a_1^{(k)}\!\cdot\vec a_2^{(k)}\Bigr) \nonumber   \\
     & = \frac{1}{2}\bigl(\|\vec a_1^{(k)}\|^2 + \|\vec a_2^{(k)}\|^2\bigr),
  \end{align}
  where $-i/2$ multiplied by the real inner product yields a purely imaginary number, so the real part vanishes.

  Self-inner product:
  \begin{align}
    v^{(k)}\!\cdot v^{(k)}
     & = \alpha^2\|\vec a_1^{(k)}\|^2 + \beta^2\|\vec a_2^{(k)}\|^2 + 2\alpha\beta\,(\vec a_1^{(k)}\!\cdot\vec a_2^{(k)}) \nonumber \\
     & = \frac{-i}{2}\|\vec a_1^{(k)}\|^2 + \frac{i}{2}\|\vec a_2^{(k)}\|^2 + \vec a_1^{(k)}\!\cdot\vec a_2^{(k)} \nonumber         \\
     & = \vec a_1^{(k)}\!\cdot\vec a_2^{(k)} + \frac{i}{2}\bigl(\|\vec a_2^{(k)}\|^2 - \|\vec a_1^{(k)}\|^2\bigr).
  \end{align}
\end{proof}

\begin{theorem}[The Frobenius Norm of the MABK Effective Coefficient Tensor]
  \label{thm:MABK-norm-final}
  Under the construction of the iterative formula \EquRef{eq:MABK-itera}, and when the measurement directions of all parties are unit vectors, the square of the Frobenius norm of the MABK effective coefficient tensor
  \begin{equation}
    \label{eq:MABK-norm-final}
    \|\mathbf{C}^{(n)}\|_F^2 = 1
  \end{equation}
  is independent of both the number of parties $n$ and the measurement directions of all parties.
\end{theorem}

\begin{proof}
  From Theorem \LemRef{lem:norm-decomposition}, the squared norm decomposes as
  \begin{equation}
    \|\mathbf{C}^{(n)}\|_F^2 = \frac{1}{2}\prod_{k=1}^n \|v^{(k)}\|^2 + \frac{1}{2}\operatorname{Re}\!\left[\prod_{k=1}^n \bigl(v^{(k)}\!\cdot v^{(k)}\bigr)\right].
  \end{equation}

  In quantum mechanics, the measurement directions are unit vectors:
  \begin{equation}
    \|\vec a_1^{(k)}\| = \|\vec a_2^{(k)}\| = 1, \qquad \vec a_1^{(k)}\!\cdot\vec a_2^{(k)} = \cos\theta_k.
  \end{equation}
  Substituting the results from Theorem \LemRef{lem:vector-quantities}:
  \begin{align}
    \|v^{(1)}\|^2          & = 1+1 = 2,                                                       \\
    \|v^{(k)}\|^2          & = \frac{1}{2}(1+1) = 1 \quad (k\ge 2),                           \\
    v^{(1)}\!\cdot v^{(1)} & = (1-1) + 2i\cos\theta_1 = 2i\cos\theta_1,                       \\
    v^{(k)}\!\cdot v^{(k)} & = \cos\theta_k + \frac{i}{2}(1-1) = \cos\theta_k \quad (k\ge 2).
  \end{align}
  Therefore, the first term:
  \begin{equation}
    \frac{1}{2}\prod_{k=1}^n \|v^{(k)}\|^2 = \frac{1}{2}\cdot 2 \cdot 1^{n-1} = 1.
  \end{equation}
  The second term:
  \begin{equation}
    \prod_{k=1}^n \bigl(v^{(k)}\!\cdot v^{(k)}\bigr) = 2i\cos\theta_1 \cdot \prod_{k=2}^n \cos\theta_k = 2i\prod_{k=1}^n \cos\theta_k,
  \end{equation}
  which is purely imaginary, so its real part vanishes:
  \begin{equation}
    \frac{1}{2}\operatorname{Re}\!\left[\prod_{k=1}^n \bigl(v^{(k)}\!\cdot v^{(k)}\bigr)\right] = 0.
  \end{equation}
  In summary,
  \begin{equation}
    \|\mathbf{C}^{(n)}\|_F^2 = 1 + 0 = 1.
  \end{equation}
  This result is independent of both $n$ and the angles $\theta_k$ between the measurement directions of all parties.
\end{proof}

\begin{remark}\label{rem:Normalization of Frobenius}
  The above recursion uses the normalized CHSH starting point, for which the classical LHV bound equals $1$ and $\|C\|_F^2=1$. If the unnormalized CHSH operator (classical bound $2$) is used instead, all coefficients scale by factor $2$, yielding $\|C\|_F^2=4$.
\end{remark}

The results presented above for the standard MABK can be naturally generalized to the more general EMABK (Extended MABK) construction. In EMABK, the two blocks in the recursive definition can be two independent MABK operators, rather than necessarily being the $1\leftrightarrow 2$ interchange of the same operator.

\begin{definition}[EMABK Operator]
  \label{def:EMABK}
  Let $\mathcal{B}_{n-1}^{(A)}$ and $\mathcal{B}_{n-1}^{(B)}$ be two independent standard MABK operators acting on the first $n-1$ parties (which may employ identical or distinct measurement settings). The $n$-party EMABK operator is defined as
  \begin{equation}
    \label{eq:EMABK-def}
    \mathcal{B}_n^{\mathrm{EMABK}} = \frac{1}{2}\Bigl[\,\mathcal{B}_{n-1}^{(A)}\otimes\bigl(M_1^{(n)}+M_2^{(n)}\bigr) + \mathcal{B}_{n-1}^{(B)}\otimes\bigl(M_1^{(n)}-M_2^{(n)}\bigr)\Bigr].
  \end{equation}
  When $\mathcal{B}_{n-1}^{(B)} = \widetilde{\mathcal{B}}_{n-1}^{(A)}$ (i.e., the $1\leftrightarrow 2$ interchange), this reduces to the standard MABK.
\end{definition}

We now show that the Frobenius norm of the effective coefficient tensor remains constant under this extended construction.

\begin{lemma}[Expansion of the EMABK Effective Coefficient Tensor]
  \label{lem:EMABK-spatial}
  Let the effective-coefficient tensors of $\mathcal{B}_{n-1}^{(A)}$ and $\mathcal{B}_{n-1}^{(B)}$ be $C^{(A)}_{p_1\dots p_{n-1}}$ and $C^{(B)}_{p_1\dots p_{n-1}}$, respectively. Then the effective-coefficient tensor of the EMABK operator \EquRef{eq:EMABK-def} is
  \begin{equation}
    \label{eq:EMABK-spatial}
    C^{\mathrm{EMABK}}_{p_1\dots p_n} = \frac{1}{2}\Bigl(S^{(n)}_{p_n}\,C^{(A)}_{p_1\dots p_{n-1}} + D^{(n)}_{p_n}\,C^{(B)}_{p_1\dots p_{n-1}}\Bigr),
  \end{equation}
  where $S^{(n)}_{p_n} = a^{(n)}_{1p_n} + a^{(n)}_{2p_n}$, $D^{(n)}_{p_n} = a^{(n)}_{1p_n} - a^{(n)}_{2p_n}$.
\end{lemma}

\begin{proof}
  The proof is identical to that of  \LemRef{lem:spatial-recursion}, except that $\mathcal{B}_{n-1}$ is replaced by $\mathcal{B}_{n-1}^{(A)}$, $\widetilde{\mathcal{B}}_{n-1}$ is replaced by $\mathcal{B}_{n-1}^{(B)}$, and the coefficients are directly compared in the Pauli basis.
\end{proof}

\begin{lemma}[Geometric Identities of the Sum and Difference Vectors]
  \label{lem:SD-identities}
  Let $\vec a_1, \vec a_2 \in \mathbb{R}^3$ be unit vectors with angle $\theta$ between them (i.e., $\vec a_1 \cdot \vec a_2 = \cos\theta$). Define
  \begin{equation}
    \vec S = \vec a_1 + \vec a_2, \quad \vec D = \vec a_1 - \vec a_2.
  \end{equation}
  Then
  \begin{align}
    \label{eq:S-norm}
    \|\vec S\|^2        & = 2 + 2\cos\theta, \\
    \label{eq:D-norm}
    \|\vec D\|^2        & = 2 - 2\cos\theta, \\
    \label{eq:SD-orthog}
    \vec S \cdot \vec D & = 0.
  \end{align}
  In particular,
  \begin{equation}
    \label{eq:SD-sum}
    \|\vec S\|^2 + \|\vec D\|^2 = 4.
  \end{equation}
\end{lemma}

\begin{proof}
  Direct calculation yields:
  \begin{equation}
    \begin{aligned}
      \|\vec S\|^2        & = \|\vec a_1\|^2 + \|\vec a_2\|^2 + 2\vec a_1\!\cdot\!\vec a_2 \\
                          & = 1 + 1 + 2\cos\theta = 2 + 2\cos\theta,                       \\
      \|\vec D\|^2        & = \|\vec a_1\|^2 + \|\vec a_2\|^2 - 2\vec a_1\!\cdot\!\vec a_2 \\
                          & = 1 + 1 - 2\cos\theta = 2 - 2\cos\theta,                       \\
      \vec S \cdot \vec D & = (\vec a_1+\vec a_2)\cdot(\vec a_1-\vec a_2)                  \\
                          & = \|\vec a_1\|^2 - \|\vec a_2\|^2 = 1 - 1 = 0.
    \end{aligned}
  \end{equation}
  Adding the first two equations gives \EquRef{eq:SD-sum}.
\end{proof}

\begin{theorem}[The Frobenius Norm of the EMABK Effective Coefficient Tensor]
  \label{thm:EMABK-norm}
  Under the EMABK construction of Definition \ref{def:EMABK}, and when the measurement directions of all parties are unit vectors, the square of the Frobenius norm of its effective-coefficient tensor
  \begin{equation}
    \label{eq:EMABK-norm-result}
    \|\mathbf{C}^{\mathrm{EMABK}}\|_F^2 = 1
  \end{equation}
  is independent of the number of parties $n$, the measurement directions of all parties, and the specific choices of the two MABK blocks $\mathcal{B}_{n-1}^{(A)}$ and $\mathcal{B}_{n-1}^{(B)}$.
\end{theorem}

\begin{proof}
  By \LemRef{lem:EMABK-spatial}, the EMABK effective coefficient tensor is
  \begin{equation}
    C^{\mathrm{EMABK}}_{p_1\dots p_n} = \frac{1}{2}\Bigl(S^{(n)}_{p_n}\,C^{(A)}_{p_1\dots p_{n-1}} + D^{(n)}_{p_n}\,C^{(B)}_{p_1\dots p_{n-1}}\Bigr).
  \end{equation}
  Computing the square of the Frobenius norm:
  \begin{align}
     & \|\mathbf{C}^{\mathrm{EMABK}}\|_F^2
    = \sum_{p_1,\dots,p_n} \bigl(C^{\mathrm{EMABK}}_{p_1\dots p_n}\bigr)^2 \nonumber                                                                    \\
     & = \frac{1}{4}\sum_{p_n}\sum_{p_1,\dots,p_{n-1}} \Bigl(S^{(n)}_{p_n}C^{(A)}_{p_1\dots p_{n-1}} + D^{(n)}_{p_n}C^{(B)}_{p_1\dots p_{n-1}}\Bigr)^2.
  \end{align}
  Expanding the square:
  \begin{multline}
    \|\mathbf{C}^{\mathrm{EMABK}}\|_F^2 = \frac{1}{4}\sum_{p_n}\Bigl[
    (S^{(n)}_{p_n})^2 \sum_{p_1,\dots,p_{n-1}}(C^{(A)}_{p_1\dots p_{n-1}})^2\\
    + (D^{(n)}_{p_n})^2 \sum_{p_1,\dots,p_{n-1}}(C^{(B)}_{p_1\dots p_{n-1}})^2\\
    + 2S^{(n)}_{p_n}D^{(n)}_{p_n}\sum_{p_1,\dots,p_{n-1}}C^{(A)}_{p_1\dots p_{n-1}}C^{(B)}_{p_1\dots p_{n-1}}\Bigr].
  \end{multline}
  This gives
  \begin{multline}
    \label{eq:EMABK-norm-expand}
    \|\mathbf{C}^{\mathrm{EMABK}}\|_F^2 = \frac{1}{4}\Bigl[
    \|\vec S^{(n)}\|^2 \cdot \|C^{(A)}\|_F^2
    + \|\vec D^{(n)}\|^2 \cdot \|C^{(B)}\|_F^2\\
    + 2(\vec S^{(n)}\!\cdot\!\vec D^{(n)}) \cdot \langle C^{(A)}, C^{(B)} \rangle\Bigr],
  \end{multline}
  where $\langle C^{(A)}, C^{(B)} \rangle = \sum_{p_1,\dots,p_{n-1}}C^{(A)}_{p_1\dots p_{n-1}}C^{(B)}_{p_1\dots p_{n-1}}$ is the inner product of the two block tensors.

  By  \ThmRef{thm:MABK-norm-final}, the norm of a standard MABK block is always
  \begin{equation}
    \|C^{(A)}\|_F^2 = \|C^{(B)}\|_F^2 = 1.
  \end{equation}
  By \ThmRef{lem:SD-identities}, the sum and difference vectors of the last party satisfy
  \begin{equation}
    \|\vec S^{(n)}\|^2 + \|\vec D^{(n)}\|^2 = 4, \quad \vec S^{(n)}\!\cdot\!\vec D^{(n)} = 0.
  \end{equation}
  Substituting the above results into \EquRef{eq:EMABK-norm-expand}:
  \begin{align}
     & \|\mathbf{C}^{\mathrm{EMABK}}\|_F^2                                                                                                             \\
     & = \frac{1}{4}\Bigl[ \|\vec S^{(n)}\|^2 \cdot 1 + \|\vec D^{(n)}\|^2 \cdot 1 + 2 \cdot 0 \cdot \langle C^{(A)}, C^{(B)} \rangle \Bigr] \nonumber \\
     & = \frac{1}{4} \cdot 1 \cdot \bigl(\|\vec S^{(n)}\|^2 + \|\vec D^{(n)}\|^2\bigr) \nonumber
    = \frac{1}{4}  \cdot 4 \nonumber
    = 1.
  \end{align}
  This completes the proof.
\end{proof}

\begin{remark}
  An important corollary of \ThmRef{thm:EMABK-norm} is that, regardless of whether the two MABK blocks in EMABK are:
  \begin{enumerate}
    \item the standard MABK case, $\mathcal{B}_{n-1}^{(B)} = \widetilde{\mathcal{B}}_{n-1}^{(A)}$ (i.e., the $1\leftrightarrow 2$ interchange);
    \item completely different measurement settings (e.g., $\mathcal{B}_{n-1}^{(A)}$ uses settings $\{1,2\}$ while $\mathcal{B}_{n-1}^{(B)}$ uses settings $\{3,4\}$);
    \item any intermediate case (partially overlapping settings),
  \end{enumerate}
  the square of the Frobenius norm of the effective coefficient tensor is always $1$.

  This guarantees that, in the estimation of the upper bound on the quantum value based on rank constraints, the tight Tsirelson-type bound after the normalization of Remark \ref{rem:Normalization of Frobenius} can be attained as
  \begin{equation}
    \mathcal{B}_Q \le \sqrt{\sigma_1(\widetilde{T})^2 + \sigma_2(\widetilde{T})^2},
  \end{equation}
  which is applicable to EMABK inequalities, where the constant factor in the upper bound is $\|\mathbf{C}^{\mathrm{EMABK}}\|_F = 1$ obtained after normalization.
\end{remark}

\section{Rank of the Coefficient Tensor of MABK and EMABK Bell Polynomials}\label{app:C}

\begin{theorem}[Rank Constraint on the Reshaped Matrix of the Effective Coefficient Tensor]
  Consider an $n$-party Bell inequality with effective-coefficient tensor  $\mathbf{C}\in\mathbb{R}^{3\times\cdots\times 3}$, whose components are defined as
  \begin{equation}\label{eq:deftensorC}
    C_{p_1p_2\cdots p_n} = \sum_{i_1,\ldots,i_n} c_{i_1\cdots i_n}\, a^{(1)}_{i_1p_1} a^{(2)}_{i_2p_2}\cdots a^{(n)}_{i_np_n},
  \end{equation}
  where $c_{i_1\cdots i_n}$ are real coefficients, $a^{(k)}_{i_kp_k}$ is the $p_k$-th spatial component of the $i_k$-th measurement direction of the $k$-th party ($p_k\in\{x,y,z\}$), and each measurement direction is a unit vector in $\mathbb{R}^3$. If there exists at least one party, which we may take as the first party without loss of generality, whose number of measurement settings is $m_1=2$ (i.e., only two measurement directions), then the $3\times 3^{n-1}$ real matrix obtained by separating and reshaping with respect to the first index
  \begin{equation}\label{eq:reshape-app}
    \widetilde{C}_{p_1,\,(p_2\cdots p_n)} = C_{p_1p_2\cdots p_n},
  \end{equation}
  where the column index $(p_2\cdots p_n)$ is arranged in lexicographic order, satisfies the rank constraint
  \begin{equation}\label{eq:rankbound}
    \operatorname{rank}(\widetilde{C})\le 2.
  \end{equation}
\end{theorem}

\begin{proof}
  Since the first party has only two measurement settings, we can arrange its measurement vectors into a matrix
  \begin{equation}
    A = \begin{pmatrix}
      a^{(1)}_{1x} & a^{(1)}_{1y} & a^{(1)}_{1z} \\[2pt]
      a^{(1)}_{2x} & a^{(1)}_{2y} & a^{(1)}_{2z}
    \end{pmatrix}\in\mathbb{R}^{2\times 3}.
  \end{equation}
  Define another matrix $D\in\mathbb{R}^{2\times 3^{n-1}}$, whose elements are
  \begin{equation}\label{eq:defD}
    D_{i,\,(p_2\cdots p_n)} = \sum_{i_2,\ldots,i_n} c_{i\,i_2\cdots i_n}\, a^{(2)}_{i_2p_2}\cdots a^{(n)}_{i_np_n},
  \end{equation}
  where $i=1,2,$ the ordering of the column index $(p_2\cdots p_n)$ is exactly the same as in $\widetilde{C}$.

  Substituting \EquRef{eq:deftensorC} into \EquRef{eq:reshape-app}, and explicitly writing out the summation index $i_1$ (note that $i_1\in\{1,2\}$):
  \begin{equation}
    \begin{aligned}
      \widetilde{C}_{p_1,\,(p_2\cdots p_n)}
       & = \sum_{i_1=1}^{2} a^{(1)}_{i_1p_1} \Bigg( \sum_{i_2,\ldots,i_n} c_{i_1i_2\cdots i_n}\, a^{(2)}_{i_2p_2}\cdots a^{(n)}_{i_np_n} \Bigg) \\
       & = \sum_{i=1}^{2} (A^{\mathsf{T}})_{p_1,i} \, D_{i,\,(p_2\cdots p_n)},
    \end{aligned}
  \end{equation}
  where $(A^{\mathsf{T}})_{p_1,i} = A_{i,p_1} = a^{(1)}_{i p_1}$. This shows that
  \begin{equation}
    \widetilde{C} = A^{\mathsf{T}} D.
  \end{equation}
  By the rank property of matrix products (the rank of a product does not exceed that of any factor matrix), we obtain
  \begin{equation}
    \operatorname{rank}(\widetilde{C}) \le \min\{\operatorname{rank}(A^{\mathsf{T}}), \operatorname{rank}(D)\} \le \operatorname{rank}(A).
  \end{equation}
  Since $A$ is a $2\times 3$ matrix, its row rank (or column rank) is at most 2, hence $\operatorname{rank}(\widetilde{C})\le 2$.

\end{proof}

\begin{remark}[Explanation of the General Case]
  In the proof above, we assumed that "the party with only 2 settings is the first party." If this special party is not the first party, we simply renumber all $n$ parties and label it as the first party. This operation is equivalent to performing a permutation of the tensor indices, and the rank of the reshaped matrix only undergoes a rearrangement of its rows (or columns) under this permutation, so the rank remains unchanged. The permuted form is also a legitimate Bell polynomial. Therefore, as long as the Bell inequality contains at least one party with only 2 measurement settings, regardless of its original position, the rank of the matrix obtained by reshaping the effective coefficient tensor with respect to that party's index does not exceed 2.

  In particular, the MABK inequalities studied in this paper and their generalization (EMABK) both satisfy this condition---in the recursive construction, the party that is combined last always has only two measurement settings, so the rank constraint always holds.
\end{remark}

\section{The xy-Plane Structure of MABK Operators}\label{app:emabk-anti}

\begin{theorem}
  When all measurement directions lie within the $xy$ plane and satisfy two observables are mutually orthogonal $\mathbf{a}_k\perp\mathbf{a}'_k$, the matrix of the $n$-party MABK operator $B_n$ in the computational basis contains only anti-diagonal entries. Its non-zero matrix elements are $(B_n)_{1,2^n} = c_n$ and $(B_n)_{2^n,1} = c_n^*$, with
  \begin{equation}
    B_n = 2^{\frac{n-1}{2}}\left(e^{i\alpha}|0\rangle^{\otimes n}\langle 1|^{\otimes n} + e^{-i\alpha}|1\rangle^{\otimes n}\langle 0|^{\otimes n}\right),
  \end{equation}
  where $|c_n|=2^{(n-1)/2}$, and the phase $\alpha=-\bigl(\sum_{j=1}^n\theta_j+\frac{(n-1)\pi}{4}\bigr)$ is determined by the azimuthal angles $\theta_j$ of the individual particles.
\end{theorem}

\begin{proof}
  For $n$ spin-$1/2$ particles, denote the two measurement directions on the $k$-th particle as unit vectors $\mathbf{a}_k,\mathbf{a}'_k$, with corresponding measurement operators $A_k=\mathbf{a}_k\cdot\boldsymbol{\sigma},\;A'_k=\mathbf{a}'_k\cdot\boldsymbol{\sigma}$. A standard recursive definition of the MABK Bell operator is
  \begin{align}
    B_1  & = A_1 , \qquad B'_1 = A'_1 ,                                                             \\[2mm]
    B_k  & = \frac{1}{2}\Big[B_{k-1}\otimes (A_k+A'_k) + B'_{k-1}\otimes (A_k-A'_k)\Big], \nonumber \\
    B'_k & = \frac{1}{2}\Big[B'_{k-1}\otimes (A_k+A'_k) - B_{k-1}\otimes (A_k-A'_k)\Big].\nonumber
  \end{align}
  The final inequality corresponds to the classical (local realistic) upper bound of $\langle B_n\rangle$.

  In the computational basis $\{|0\rangle,|1\rangle\}$, introduce the operators $\sigma_{+}=|0\rangle\langle 1|,\;\sigma_{-}=|1\rangle\langle 0|$. The Pauli matrices can be written as
  \begin{equation}
    X = \sigma_{+} + \sigma_{-},\quad Y = i(\sigma_{-} - \sigma_{+}).
  \end{equation}
  Choose all measurement directions to lie in the $xy$ plane, i.e.,
  \begin{equation}
    \begin{aligned}
      \mathbf{a}_k  & = (\cos\theta_k,\;\sin\theta_k,\;0),                                             \\
      \mathbf{a}'_k & = \bigl(\cos(\theta_k+\tfrac{\pi}{2}),\;\sin(\theta_k+\tfrac{\pi}{2}),\;0\bigr),
    \end{aligned}
  \end{equation}
  meaning that the two measurement directions of each particle are mutually orthogonal (differing by $90^\circ$). Then we have
  \begin{equation}
    \begin{aligned}
      A_k  & = e^{-i\theta_k}\sigma_{+} + e^{i\theta_k}\sigma_{-} ,      \\
      A'_k & = -i e^{-i\theta_k}\sigma_{+} + i e^{i\theta_k}\sigma_{-} .
    \end{aligned}
  \end{equation}
  To simplify notation, let
  \begin{equation}
    \begin{aligned}
      A_k+A'_k & = u\,\sigma_{+} + u^*\sigma_{-},\qquad u = \sqrt2\,e^{-i(\theta_k+\pi/4)}, \\
      A_k-A'_k & = v\,\sigma_{+} + v^*\sigma_{-},\qquad v = \sqrt2\,e^{-i(\theta_k-\pi/4)}.
    \end{aligned}
  \end{equation}
  Clearly $u^*/v^* = i$.

  \begin{theorem}\label{thm:antidiag}
    For all $n$, under the chosen measurement settings, $B_n$ and $B'_n$ contain only $\sigma_{+}^{\otimes n}$ and $\sigma_{-}^{\otimes n}$, i.e., they can be written as
    \begin{equation}
      B_n = c_n\,\sigma_{+}^{\otimes n} + c_n^*\,\sigma_{-}^{\otimes n},\qquad
      B'_n = d_n\,\sigma_{+}^{\otimes n} + d_n^*\,\sigma_{-}^{\otimes n},
    \end{equation}
    i.e., only anti-diagonal entries.
  \end{theorem}

  \begin{proof}
    Direct calculation from the definition gives:
    \begin{equation}
      B_1 = e^{-i\theta_1}\sigma_{+} + e^{i\theta_1}\sigma_{-} ,\quad
      B'_1 = -i e^{-i\theta_1}\sigma_{+} + i e^{i\theta_1}\sigma_{-} ,
    \end{equation}
    so $c_1 = e^{-i\theta_1},\;d_1 = -i c_1$, and the condition holds.

    Assume the conclusion holds for $k-1$, and $d_{k-1}=-i c_{k-1}$. Applying the recursion for the $k$-th particle:
    \begin{equation}
      \begin{aligned}
        B_k  & = \frac12\Big[B_{k-1}\otimes (u\sigma_{+}+u^*\sigma_{-}) + B'_{k-1}\otimes (v\sigma_{+}+v^*\sigma_{-})\Big], \\
        B'_k & = \frac12\Big[B'_{k-1}\otimes (u\sigma_{+}+u^*\sigma_{-}) - B_{k-1}\otimes (v\sigma_{+}+v^*\sigma_{-})\Big].
      \end{aligned}
    \end{equation}
    After expansion, the coefficients of the mixed terms such as $\sigma_{+}^{\otimes(k-1)}\!\otimes\sigma_{-}$ (and its Hermitian conjugate) are respectively
    \begin{equation}
      B_k:\; c_{k-1}u^* + d_{k-1}v^*,\qquad
      B'_k:\; d_{k-1}u^* - c_{k-1}v^*.
    \end{equation}
    Using $d_{k-1}=-i c_{k-1}$ and $u^*/v^* = i$, we obtain
    \begin{equation}
      c_{k-1}u^* + d_{k-1}v^* = c_{k-1}(u^* - i v^*) = 0,
    \end{equation}
    and similarly the other mixed term also vanishes. Therefore all mixed terms cancel exactly, leaving only $\sigma_{+}^{\otimes k}$ and $\sigma_{-}^{\otimes k}$.

    The coefficients of the pure terms are then
    \begin{equation}
      c_k = \frac12\big(c_{k-1}u + d_{k-1}v\big),\qquad
      d_k = \frac12\big(d_{k-1}u - c_{k-1}v\big).
    \end{equation}
    Substituting $d_{k-1}=-i c_{k-1}$ and simplifying:
    \begin{align}
      c_k & = \frac12 c_{k-1}(u - i v),                                                      \\
      d_k & = \frac12 c_{k-1}(-i u - v) = -i\cdot\frac12 c_{k-1}(u - i v) = -i c_k.\nonumber
    \end{align}
    Thus the phase matching condition is preserved at step $k$. The induction is complete.

    Since $\sigma_{+}^{\otimes n}=|0\cdots0\rangle\langle1\cdots1|$ and $\sigma_{-}^{\otimes n}=|1\cdots1\rangle\langle0\cdots0|$,
    this means that in the computational basis, the matrix of $B_n$ has non-zero elements only at the two endpoints $(1,2^n)$ and $(2^n,1)$ of the anti-diagonal, with all other entries being zero.
  \end{proof}

  \begin{lemma}\label{lem:coeff}
    From the recurrence relation $c_k = \frac12 c_{k-1}(u - i v)$, substituting the expressions for $u,v$:
    \begin{equation}
      u - i v = \sqrt2\,e^{-i\theta_k}\bigl(e^{-i\pi/4} - i e^{i\pi/4}\bigr) = 2\sqrt2\,e^{-i(\theta_k+\pi/4)},
    \end{equation}
    so
    \begin{equation}
      c_k = \sqrt2\,e^{-i(\theta_k+\pi/4)}c_{k-1},\quad c_1 = e^{-i\theta_1}.
    \end{equation}
    Iterating yields
    \begin{equation}
      c_n = 2^{\frac{n-1}{2}} \exp\!\Bigl[-i\Bigl(\sum_{j=1}^n\theta_j + \frac{(n-1)\pi}{4}\Bigr)\Bigr].
    \end{equation}
    Therefore, the only non-zero matrix elements of $B_n$ are
    \begin{equation}
      (B_n)_{1,2^n} = c_n,\qquad (B_n)_{2^n,1} = c_n^*,
    \end{equation}
    with modulus always
    \begin{equation}
      |c_n| = 2^{\frac{n-1}{2}}.
    \end{equation}
  \end{lemma}
  From the above lemma, under the chosen measurement settings (all $\mathbf{a}_k\perp\mathbf{a}'_k$ lying in the $xy$ plane),
  \begin{equation}
    B_n = 2^{\frac{n-1}{2}}\bigl(e^{i\alpha}|0\cdots0\rangle\langle1\cdots1| + e^{-i\alpha}|1\cdots1\rangle\langle0\cdots0|\bigr),
  \end{equation}
  where the phase $\alpha = -\bigl(\sum_j\theta_j + \frac{(n-1)\pi}{4}\bigr)$ varies continuously with the measurement azimuthal angles.
\end{proof}

\end{document}